%% file: paper.tex
\documentclass{sig-alternate}
\pdfoutput=1

\usepackage[numbers]{natbib}
\makeatletter
\newif\if@restonecol
\makeatother

\let\chapter\section
\usepackage[linesnumbered,boxed,commentsnumbered,ruled]{algorithm2e}

\usepackage{framed}
\usepackage{color, colortbl}
\newtheorem{problem}{Problem}
\newtheorem{lemma}{Lemma}

\newcommand{\hh}[1]{{\small\color{red}{\bf hh: #1}}}
\newcommand{\lli}[1]{{\small\color{blue}{\bf lli: #1}}}

\newcommand{\hide}[1]{}

\newcommand{\mkclean}{
   \renewcommand{\hh}[1]{}
   \renewcommand{\lli}[1]{}
}

\newcommand{\mat}[1]{{\bf #1}}   
\newcommand{\teamrep}{\textsc{Team Member Replacement} }  %
\newcommand{\teamrepalgbasic}{\textsc{TeamRep-Basic}}  %
\newcommand{\teamrepalgexact}{\textsc{TeamRep-Fast-Exact}}  %
\newcommand{\teamrepalgapp}{\textsc{TeamRep-Fast-Approx}}  %

\usepackage{balance}
\usepackage{caption}
\usepackage{subcaption}
\usepackage[fleqn,tbtags]{mathtools}

\begin{document}
%
\conferenceinfo{xxx}{2014, xxx}

\title{Replacing the Irreplaceable: Fast Algorithms for Team Member Recommendation}

%
%
%
%

\numberofauthors{6} 
%

\author{
%
%
\alignauthor
Liangyue Li\\
       \affaddr{Arizona State University}\\
       \email{liangyue@asu.edu}
\alignauthor
Hanghang Tong\\
       \affaddr{Arizona State University}\\
       \email{Hanghang.Tong@asu.edu}
\alignauthor
Nan Cao\\
       \affaddr{IBM Research}\\
       \email{nancao@us.ibm.com}
\and
\alignauthor
Kate Ehrlich \\
		\affaddr{IBM Research}\\
		\email{katee@us.ibm.com}
\alignauthor
Yu-Ru Lin\\
		\affaddr{University of Pittsburgh}\\
		\email{yurulin@pitt.edu}
\alignauthor
Norbou Buchler\\
		\affaddr{US Army Research Laboratory}\\
		\email{norbou.buchler.civ@mail.mil}		
}


\maketitle
\mkclean
\begin{abstract}

\input{000abstract.tex}

\end{abstract}




\section{Introduction}
\input{001intro.tex}
\section{Problem Definitions}

\input{003probdef.tex}

\section{Proposed Solutions}
\input{004preliminary.tex}

\section{Scale-up and Speed-up}
\input{005method.tex}

\vspace{-10pt}
\section{Experimental Evaluations}
\input{006exp.tex}

\vspace{-10pt}
\section{Related Work}\label{sec:rel}
\input{002related.tex}

\vspace{-10pt}

\section{Conclusion}
\input{007con.tex}

\vspace{-18pt}

%
\small
\bibliographystyle{plainnat}
\bibliography{sigproc} 
 
\end{document}

%% file: 000abstract.tex
In this paper, we study the problem of \teamrep: given a team of people embedded in a social network working on the same task, find a good candidate who can fit in the team after one team member becomes unavailable. We conjecture that a good team member replacement  should have good \emph{skill matching} as well as good \emph{structure matching}.  We formulate this problem using the concept of graph kernel. To tackle the computational challenges, we propose a family of fast algorithms by (a) designing effective pruning strategies, and (b) exploring the smoothness between the existing and the new team structures. We conduct extensive experimental evaluations on real world datasets to demonstrate the effectiveness and efficiency. Our algorithms (a) perform significantly better than the alternative choices in terms of both precision and recall; and (b) scale {\em sub-linearly}.

%% file: 001intro.tex
In his world-widely renowned book ``The Science of the Artificial''~\cite{Simon:1996:SA:237774}, Nobel laureate Herbert Simon pointed out that it is more the complexity of the {\em environment}, than the complexity of the individual persons, that determines the complex behavior of humans. The emergence of online social network sites and web 2.0 applications provides a new connected environment/context, where people interact and collaborate with each other as a team to collectively perform some complex tasks. 

Among others, the churn of team members is a common problem across many application domains. To name a few, an employee in a software or sales team might decide to leave the organization and/or be assigned to a new task. In a law enforcement mission, a SWAT team might lose certain task force due to the fatality or injury. In professional sports (e.g., NBA), the rotation tactic between the benches could play a key role on the game outcome. 
In all these cases, the loss of the key member (i.e., the irreplaceable) might bring the catastrophic consequence to the team performance. How can we find the best alternate when a team member becomes unavailable? However, despite the frequency with which people leave a team before a project/task is complete and the resulting disruption~\cite{DBLP:conf/chi/ZadehBKC11}, replacements are often found opportunistically and are not necessarily optimal.


We conjecture there will be less disruption when the team member who leaves is replaced with someone with similar relationships with the other team members.   This conjecture is inspired by some recent research which shows that team members prefer to work with people they have worked with before~\cite{Hinds00choosingwork} and that distributed teams perform better when members know each other~\cite{DBLP:conf/cscw/CummingsK08}. Furthermore, research has shown that specific communication patterns amongst team members are critical for performance~\cite{DBLP:conf/chi/CataldoE12}. Thus, in addition to factors such as skill level, maintaining the same or better level of familiarity and communication amongst team members before and after someone leaves should reduce the impact of the departure. In other words, for team member replacement, the similarity between individuals should be measured in the context of the team itself. More specially, a good team member replacement should meet the following two requirements. First ({\em skill matching}), the new member should bring a similar skill set as the current team member to be replaced. Second ({\em structure matching}), the new member should have a similar network structure as the current team member in connecting the rest team members. 

\begin{figure*}[!htb]
\centering
\includegraphics[width=\textwidth]{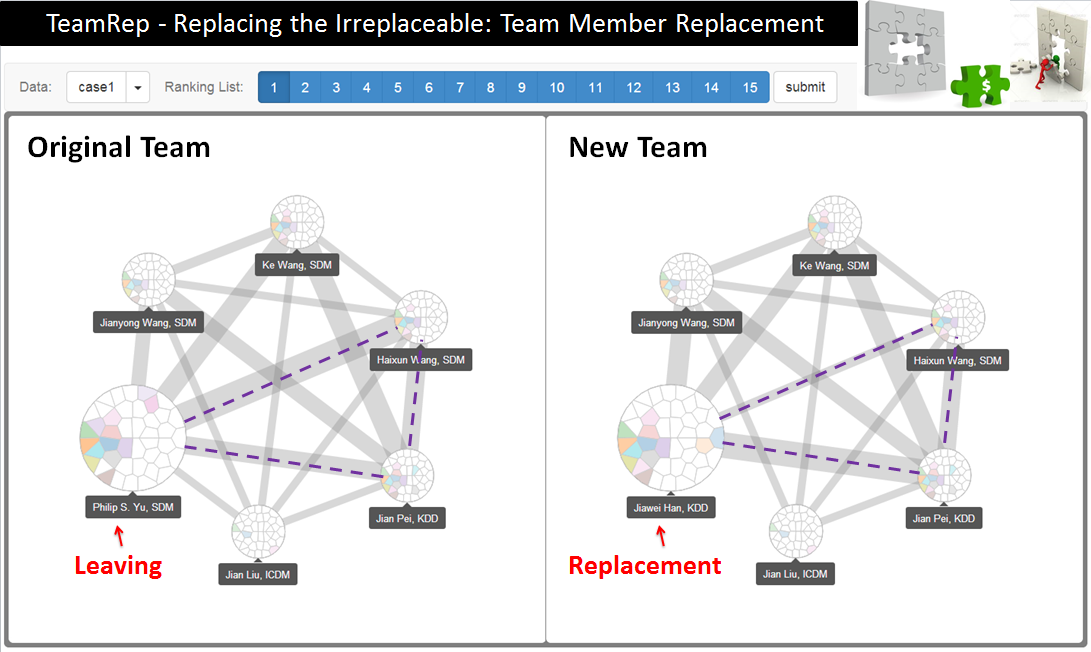}
\vspace{-10pt}
\caption{The team graphs before and after {\em Jiawei Han} takes {\em Philip S. Yu}'s place. See subsection 5.2 for detailed explanations.}
\label{fig:example}
\vspace{-10pt}
\end{figure*}

Armed with this conjecture, we formally formulate the \teamrep problem by the notation of graph similarity/kernel. By modeling the team as a labeled graph, graph kernel/similarity provides a natural way to capture both the skill and structural match.  However, for a network with $n$ individuals, a straight-forward method would require $O(n)$ graph kernel computations for one team member replacement, which is computationally intractable. For example, for the {\em DBLP} dataset with almost 1M users (i.e., $n\approx 1,000,000$), we found that it would take 6,388s to find one replacement for a team of size 10. To address the computational challenges, we propose a family of fast algorithms by carefully designing the pruning strategies and exploring the smoothness between the existing and the new teams. We perform the extensive experimental evaluations to demonstrate the effectiveness and efficiency of our methods. Specially, we find that (1) by encoding both the skill and structural matching, it leads to a much better replacement result. Compared with the best alternative choices, we achieve 27\% and 24\% {\em net increase} in average recall and precision, respectively; (2) our fast algorithms are orders of magnitude faster and scale {\em sub-linearly}. For example, our pruning strategy alone leads up to 1,709$\times$ speed-up, without sacrificing any accuracy.

The main contributions of this paper are as follows.
\begin{itemize}
\setlength{\itemsep}{0.05\baselineskip}
\item [1] \textbf{Problem Definitions.} We formally define the \teamrep\ problem, to recommend a good candidate after a team member is unavailable in the context of networks.

\item [2] \textbf{Algorithms and Analysis.} We propose a family of effective and scalable algorithms for \teamrep; and analyze its correctness and complexity.

\item [3] {\bf Experimental Evaluations}. We perform extensive experiments on real world datasets, to validate the effectiveness and efficiency of our methods. (See an example in Figure~\ref{fig:example}.)
\end{itemize}

The rest of the paper is organized as follows.  Section 2 formally defines \teamrep\ problem. Section 3 presents our basic solutions. Section 4 addresses the computational challenges. Section 5 presents the experimental results. Section~\ref{sec:rel} reviews some related work. Section 7 concludes the paper.

%% file: 003probdef.tex
\begin{table}[!t]
\caption{Table of symbols}
\centering
\vspace{-8pt}
\begin{tabular}{|c|l|}
\hline
\bf{Symbols}                & \bf{Definition}\\
\hline \hline
$\mat{G}:=\{\mat{A}, \mat{L} \}$ & the entire social network \\
$\mat{A}_{n \times n}$ & the adjacency matrix of $\mat{G}$\\
$\mat{L}_{n \times l}$ & skill indicator matrix\\ \hline
$\mathcal{T}$               & the team member index \\
$\mat{G}(\mathcal{T})$ & the team network indexed by its members $\mathcal{T}$\\\hline
$d_i$				   & the degree of the $i^{\textrm{th}}$ node in $\mat A$\\
$l$ &the total number of skills \\
$t$                  & the team size, i.e., $t = |\mathcal{T}|$\\
$n$                  & the total number of individuals in $\mat A$\\
$m$                  & the total number of connections in $\mat A$ \\ \hline
\end{tabular}
\label{tab:symbol}
\end{table}

\hide{
\begin{figure*}[htp]
\centering
\begin{subfigure}[b]{0.45\textwidth}
\includegraphics[width=\textwidth]{figures/structure-exp-1}
\caption{Original graph}
\label{fig:structure_before}
\end{subfigure}
~
\begin{subfigure}[b]{0.45\textwidth}
\includegraphics[width=\textwidth]{figures/structure-exp-2}
\caption{Graph after replacement}
\label{fig:structure_after}
\end{subfigure}
\caption{Visualization of team graphs before and after Jiawei Han takes Philip S. Yu's place.}
\label{fig:structure}
\end{figure*}
}


Table~\ref{tab:symbol} lists the main symbols used throughout this paper. We describe the $n$ individuals by an labelled social network $\mat{G}:=\{\mat{A}, \mat{L} \}$, where $\mat{A}$ is an $n \times n$ adjacency matrix characterizing the connectivity among different individuals; and $\mat L$ is $n \times l$ skill indicator matrix. The $i^{\textrm{th}}$ row vector of $\mat L$ describes the skill set of the $i^{\textrm{th}}$ individual. For example, suppose there are only three skills in total, including \{{\em data mining, databases, information retrieval}\}. Then an individual with a skill vector $[1,1,0]$ means that s/he has both {\em data mining} and {\em databases} skills but no skill in terms of {\em information retrieval}. Also, we represent
the elements in a matrix using a convention similar to Matlab, e.g.,
$\mat{A}(i,j)$ is the element at the $i^{\mathrm{th}}$ row and
$j^{\mathrm{th}}$ column of the matrix $\mat{A}$, and $\mat{A}(:,j)$
is the $j^{\mathrm{th}}$ column of $\mat{A}$, etc.

We use the calligraphic letter $\mathcal{T}$ to index the members of a team, which includes a subset of $t = |\mathcal{T}|$ out-of $n$ individuals. Correspondingly, we can represent the team by another labelled team network $\mat{G}(\mathcal{T}):=\{\mat{A}(\mathcal{T}, \mathcal{T}), \mat{L}(\mathcal{T},:)\}$. Note that $\mat{A}(\mathcal{T}, \mathcal{T})$ and $\mat{L}(\mathcal{T},:)\}$ are sub-matrices of $\mat A$ and $\mat L$, respectively. If we replace an existing member $p \in \mathcal{T}$ of a given team $\mathcal{T}$ by another individual $q \notin \mathcal{T}$, the new team members are indexed by $\mathcal{T}_{p\rightarrow q}:=\{\mathcal{T}/p, q\}$; and the new team is represented by the  labelled network $\mat{G}(\mathcal{T}_{p \rightarrow q})$.

With the above notations and assumptions, our problems can be
formally defined as follows:

\begin{problem} {\teamrep}\label{prob:def1}
\begin{description}
\itemsep -2.5pt
\item[Given:] (1) A labelled social network $\mat{G}:=\{\mat{A}, \mat{L} \}$, (2) a team $\mat{G}(\mathcal T)$, and (3) a team member $p \in \mathcal{T}$;
\item[Output:] A ``best" alternate $q \notin \mathcal{T}$ to replace the person $p$'s role in the team $\mat{G}(\mathcal T)$.
\end{description}
\end{problem}

%% file: 004preliminary.tex



In this section, we present our solution for Problem~\ref{prob:def1}. We start with the design objectives for the \teamrep\ problem, present graph kernel as the basic solution to fulfill such design objectives; and finally analyze the main computational challenges.

\subsection{Design Objectives}
Generally speaking, we want to find a {\em similar} person $q$ to replace the current team member $p$ who is about to leave the team. That is, a good replacement $q$ should not only have a similar skill set as team member $p$; but also would maintain the good chemistry of the team so that the whole team can work together harmonically and/or be less disrupted. In other words, the similarity between individuals should be measured in the context of the team itself. Often, the success of a team largely depends on the successful execution of several sub-tasks, each of which requires the cooperation among several team members with certain skill configurations. For example, several classic tactics often recurringly find themselves in a successful NBA team, including (a) {\em triangle offense} (which is featured by a sideline triangle created by {\em the center}, {\em the forward}, and {\em the guard}), (b) {\em pick and roll} (which involves the cooperation between two players - one plays as `pivot' and the other plays as `screen', respectively), etc. Generally speaking, team performance arises from the shared knowledge and experience amongst team members and their ability to share and coordinate their work. As noted in the introduction, a specific pattern of communication, is associated with higher team performance.  Maintaining that communication structure should therefore be less disruptive to the team.

If we translate these requirements into the notations defined in Section 2, it naturally leads to the following two design objectives for a good \teamrep. First ({\em skill matching}), the new member should bring a similar skill set as the current team member $p$ to be replaced that are required by the team. Second ({\em structural matching}), the new member should have a similar network structure as team member $p$ in connecting the rest team members.

\subsection{Basic Solutions}

In order to fulfill the above two design objectives, we need a similarity measure between two individuals in the context of the team itself, that captures both skill matching and the structural matching. We refer to this kind of similarity as {\em team context aware similarity}. Mathematically, the so-called graph kernel defined on the current and new teams provides a natural tool for such a team context aware similarity. That is, we want to find a replacement person $q$ as
\begin{eqnarray}\label{eq:teamsim}
q = \textrm{argmax}_{j,j\notin {\mathcal{T}}}~~\textrm{Ker}(\mat{G}({\cal {T}}),\mat{G}(\mathcal{T}_{p \rightarrow j}))
\end{eqnarray}
In Eq.~\eqref{eq:teamsim}, $\mat{G}({\cal {T}})$ is the labelled team graph; and $\mat{G}(\mathcal{T}_{p \rightarrow j})$ is the labelled team graph after we replace a team member $p$ by another individual $j$; and $\textrm{Ker}(~.~)$ is the kernel between these two labelled graphs. Generally speaking, the basic idea of various graph kernels is to compare the similarity of the {\em sub-graphs} between the two input graphs and then aggregate them as the overall similarity between the two graphs. Let us explain the intuition/rationality of why graph kernel is a natural choice for team context aware similarity. Here, each subgraph in a given team (e.g., the dashed triangles in Figure~\ref{fig:example}) might reflect a specific skill configuration among a sub-group of team members that is required by a certain sub-task of that team. By comparing the similarity between two subgraphs, we implicitly measure the capability of the individual $j$ to perform this specific sub-task. Thus, by aggregating the similarities of all the possible subgraphs between the two input graphs/teams, we get a goodness measure of the overall capability of the individual $j$ to perform all the potential sub-tasks that team member $p$ is involved in the original team. Note that the team replacement scenario is different from team formation~\cite{DBLP:conf/kdd/LappasLT09,DBLP:conf/www/AnagnostopoulosBCGL12,DBLP:conf/www/RangapuramBH13}. These existing work on team formation aims to build a team from scratch by optimizing some pre-chosen metric (e.g., compatibility, diversity, etc). In contract, we aim to find a new team member such that the new team resembles the original team as much as possible.

Having this in mind, many of the existing graph kernels can be adopted in Eq.~\eqref{eq:teamsim}, such as random walk based graph kernel, sub-tree based graph kernels (See section~\ref{sec:rel} for a review). In this paper, we focus on random walk based graph kernel due to its mathematical elegancy and superior empirical performance. Given two labelled graphs $\mat{G}_i:=\{\mat{A}_i,\mat{L}_i\}$, $i=1,2$, the random walk based graph kernel between them can be formally computed as follows~\cite{DBLP:conf/nips/VishwanathanBS06}:
\begin{eqnarray}
\textrm{Ker}(\mat{G}_1,\mat{G}_2) = \mat{y}' (\mat{I} - c \mat{A}_{\times}) ^{-1}\mat{L}_{\times} \mat{x}
\label{eq:original}
\end{eqnarray}
where $\mat{A}_{\times}=\mat{L}_{\times} (\mat{A}'_1 \otimes \mat{A}'_2)$ is the weight matrix of the two graphs' Kronecker product, $\otimes$ represents the Kronecker product between two matrices, $c$ is a decay factor, $\mat{y}=\mat{y}_1 \otimes \mat{y}_2$ and $\mat{x}=\mat{x}_1 \otimes \mat{x}_2$ are the so-called starting and stopping vectors to indicate the weights of different nodes and are set uniform in our case, $\mat{L}_{\times}$ is a diagonal matrix where $\mat{L}_{\times}(i,i) = 0$ if the $i^{\mathrm{th}}$ row of $(\mat{A}'_1 \otimes \mat{A}'_2)$ is zeroed out due to label inconsistency of two nodes of the two graphs. $\mat{L}_{\times}$ can be expressed as $\mat{L}_{\times} = \sum_{k=1}^{l} \mathrm{diag}(\mat{L}_1(:,k)) \otimes \mathrm{diag}(\mat{L}_2(:,k))$.


\subsection{Computational Challenges}

Eq.\eqref{eq:original} naturally suggests the following procedure for solving \teamrep\ problem (referred to as \teamrepalgbasic): for each individual $j\notin \mathcal{T}$, we compute its score $\textrm{score}(j)$ by Eq.\eqref{eq:original}; and recommend the individual(s) with the highest score(s). However, this strategy (\teamrepalgbasic) is computationally intensive since we need to compute {\em many} random walk based graph kernels and each of such computations could be expensive especially when the team size is large. To be specific, for a team $\mathcal T$ of size $t$ and a graph $\mat G$ with $n$ individuals in total, its time complexity is $O(nt^6)$ since we need to compute a random walk based graph kernel for each candidate who is not in the current team, each of which could cost $O(t^6)$~\cite{DBLP:conf/nips/VishwanathanBS06} \hh{Liangyue: could you please double-check Vish's NIPS 06 paper as well as sec 2.2 of u kang's paper - i think for exact method, they have cheaper algorithms that are better than $O(t^6)$}. Even if we allow some approximation in computing each of these graph kernels, the best known algorithms (i.e., by ~\cite{kang:fastwalk}) would still give an overall time complexity as  $O(n(lt^2r^4 + mr + r^6))$, where $r$ is reduced rank after low rank approximation, \hh{liangyue: pls fill in} which is still too high. For example, on the {\em DBLP} dataset with 916,978 authors, for a team with 10 members, it would take 6,388s to find a best replacement.

In the next section, we present our solution to remedy these computational challenges.

%% file: 005method.tex
In this section, we address the computational challenges to scale-up and speed-up \teamrepalgbasic. We start with an efficient pruning strategy to reduce the number of graph kernel computations; and then present two algorithms to speed-up the computation of individual graph kernel. 

\subsection{Scale-up: Candidate Filtering}

Here, we propose an efficient pruning strategy to filter out those unpromising candidates. Recall that one of our design objectives for a good \teamrep is {\em structural matching}, i.e., the new member has a similar network structure as team member $p$ in connecting the rest team members. Since $p$ is connected to at least some of the rest members, it suggests that if an individual does not have any connection to any of the rest team members, s/he might not be a good candidate for replacement.

{\em Pruning Strategy:} Filter out all the candidates who do not have any connections to any of the rest team members.

\begin{lemma}\label{lm:pruning}
{\em Effectiveness of Pruning}. For any two persons $i$ and $j$ not in $\mathcal{T}$,  if $i$ is connected to at least one member in $\mathcal{T}/p$ and $j$ has no connections to any of the members in $\mathcal{T}/p$, we have that
\begin{eqnarray}
\textrm{Ker}(\mat{G}({\cal {T}}),\mat{G}(\mathcal{T}_{p \rightarrow i})) \geq \textrm{Ker}(\mat{G}({\cal {T}}),\mat{G}(\mathcal{T}_{p \rightarrow j})). \nonumber
\end{eqnarray}
\end{lemma}

\begin{proof}
Suppose that $\mat{G}(\mathcal{T}):=\{\mat{A}_0, \mat{L}_0\}$. Let $\mat{G}(\mathcal{T}_{p \rightarrow i}):=\{\mat{A}_1, \mat{L}_1\}$, and $\mat{G}(\mathcal{T}_{p \rightarrow j}):=\{\mat{A}_2, \mat{L}_2\}$.

By Taylor expansion of Eq.~\eqref{eq:original}, we have

$\textrm{Ker}(\mat{G}({\cal {T}}),\mat{G}(\mathcal{T}_{p \rightarrow i})) = \sum_{z=0}^{\infty} c \mat{y}' (\mat{L}_{\times 1}(\mat{A}'_0 \otimes \mat{A}'_1))^z \mat{x}$, where $\mat{L}_{\times 1} = \sum_{k=1}^{l} \mathrm{diag}(\mat{L}_0(:,k)) \otimes \mathrm{diag}(\mat{L}_1(:,k))$,

$\textrm{Ker}(\mat{G}({\cal {T}}),\mat{G}(\mathcal{T}_{p \rightarrow j})) = \sum_{z=0}^{\infty} c \mat{y}' (\mat{L}_{\times 2}(\mat{A}'_0 \otimes \mat{A}'_2))^z \mat{x}$, where $\mat{L}_{\times 2} = \sum_{k=1}^{l} \mathrm{diag}(\mat{L}_0(:,k)) \otimes \mathrm{diag}(\mat{L}_2(:,k))$.

Therefore, it is sufficient to show that $(\mat{L}_{\times 1}(\mat{A}'_0\otimes \mat{A}'_1))^z \ge (\mat{L}_{\times 2}(\mat{A}'_0 \otimes \mat{A}'_2))^z$ for any $z>0$, where two matrices $\mat{A} \ge \mat{B} $ if  $\mat{A}_{ij} \ge \mat{B}_{ij}$ holds for all possible $(i,j)$. We prove this by induction.

(Base Case of Induction) When $z=1$, we have
\begin{equation}
\begin{array}{l}
\mat{L}_{\times 1}(\mat{A}'_0 \otimes \mat{A}'_1)\\
=(\sum_{k=1}^{l} \mathrm{diag}(\mat{L}_0(:,k)) \otimes \mathrm{diag}(\mat{L}_1(:,k)))(\mat{A}'_0 \otimes \mat{A}'_1)\\
=\sum_{k=1}^{l} (\mathrm{diag}(\mat{L}_0(:,k)) \mat{A}'_0) \otimes (\mathrm{diag}(\mat{L}_1(:,k)) \mat{A}'_1)
\end{array}
\end{equation}

Because  $(\mathrm{diag}(\mat{L}_1(:,k)) \mat{A}'_1)\geq (\mathrm{diag}(\mat{L}_2(:,k)) \mat{A}'_2)$, we have $\mat{L}_{\times 1}(\mat{A}'_0 \otimes \mat{A}'_1)\ge \mat{L}_{\times 2}(\mat{A}'_0 \otimes \mat{A}'_2)$.

(Induction Step) Assuming $(\mat{L}_{\times 1}(\mat{A}'_0 \otimes \mat{A}'_1))^{z-1} \geq (\mat{L}_{\times 2}(\mat{A}'_0 \otimes \mat{A}'_2))^{z-1}$, we have that
\begin{eqnarray}
(\mat{L}_{\times 1}(\mat{A}'_0 \otimes \mat{A}'_1))^{z} &\geq& (\mat{L}_{\times 2}(\mat{A}'_0 \otimes \mat{A}'_2))^{z-1} (\mat{L}_{\times 1}(\mat{A}'_0 \otimes \mat{A}'_1)) \nonumber \\ &\geq& (\mat{L}_{\times 2}(\mat{A}'_0 \otimes \mat{A}'_2))^{z} \nonumber
\end{eqnarray}
\hh{let use prime for matrix transpose}
\nonumber where the first inequality is due to the induction~assumption; and the second inequality is due to the base case. This completes the proof.
\end{proof}
\vspace{-8pt}

{\em Remarks}. By Lemma~\ref{lm:pruning}, our pruning strategy is `safe', i.e., it will not miss any potentially good replacements. In the meanwhile,  we can reduce the number of graph kernel compuations from $O(n)$ to $O(\sum_{i\in \mathcal{T}/p}d_i)$, which is sub-linear in $n$.

\subsection{Speedup Graph Kernel - Exact Approach}


Here, we address the problem of speeding up the computation of an individual graph kernel. Let $\mat{G}(\mathcal{T}) := \{ \mat{A}_1, \mat{L}_1 \}$ and $\mat{G}(\mathcal{T}_{p \rightarrow q}) := \{ \mat{A}_2, \mat{L}_2 \}$, where $\mat{A}_1,\mat{A}_2$ are symmetric {adjacency matrices of the two graphs}.\footnote{Although we focus on the undirected graphs in this paper, our proposed algorithms can be generalized to directed graphs.} Without loss of generality, let us assume that $p$ is the last team member in $\mathcal{T}$. Compare $\mat{A}_1$ with $\mat{A}_2$, it can be seen that the only difference is their last columns and last rows. Therefore, we can rewrite $\mat{A}_2$ as $\mat{A}_2 = \mat{A}_c + \mat{A}_{d2}$, where $\mat{A}_c$ is $\mat{A}_1$ with its last row and column being zeroed out, and the nonzero elements of $\mat{A}_{d2}$  only appear in its last row and column reflecting the connectivity of $q$ to the new team. Notice that $\mat{A}_{d2}$ has a rank at most 2, so it can be factorized into two smaller matrices as $\mat{A}_{d2} = \mat{E}_{t \times 2} \mat{F}_{2 \times t}$.

Denote $\mathrm{diag}(\mat{L}_1(:,j))$ as $\mat{L}_1^{(j)}$  and $\mathrm{diag}(\mat{L}_2(:,j))$ as $\mat{L}_2^{(j)}$ for $j=1,...,l$. Compare $\mat{L}_1^{(j)}$ with $\mat{L}_2^{(j)}$, the only difference is the last diagonal element. Therefore, we can write $\mat{L}_2^{(j)}$ as $\mat{L}_2^{(j)} = \mat{L}_c^{(j)} + \mat{L}_{d2}^{(j)}$, where $\mat{L}_c^{(j)}$ is $\mat{L}_1^{(j) }$ with last element zeroed out, and $\mat{L}_{d2}^{(j)}$'s last element indicates $q$'s strength of having the $j^{\mathrm{th}}$ skill. $\mat{L}_2^{(j)}$'s rank is at most 1, so it can be factorized as $\mat{L}_2^{(j)}=\mat{e}_{t \times 1}^{(j)}\mat{f}_{1 \times t}^{(j)}$. Therefore, the exact graph kernel for the labelled graph can be computed as:
\small
\vspace{-5pt}
\begin{equation}\label{eq:faxtexact:1}
\begin{array}{l}
\mathrm{Ker}(\mat{G}(\mathcal{T}), \mat{G}(\mathcal{T}_{p\rightarrow q} ))\\
=\mat{y}' (\mat{I} - c(\sum_{j=1}^l \mat{L}_1^{(j)} \otimes \mat{L}_2^{(j)})(\mat{A}_1' \otimes \mat{A}_2') )^{-1}(\sum_{j=1}^l \mat{L}_1^{(j)} \otimes \mat{L}_2^{(j)}) \mat{x}\\
=\mat{y}' (\underbrace{\mat{I} - c(\sum_{j=1}^l \mat{L}_1^{(j)} \otimes \mat{L}_c^{(j)})(\mat{A}_1 \otimes \mat{A}_c )}_{\mat{Z}:~\textrm{invariant w.r.t.}~q~}\\
-c \underbrace{(\sum_{j=1}^l (\mat{L}_1^{(j)} \otimes \mat{e}^{(j)})(\mat{I} \otimes \mat{f}^{(j)}))(\mat{A}_1 \otimes \mat{A}_c)}_{\mat{PQ}(\mat{A}_1\otimes\mat{A}_c)=\mat{P}\mat{Y}_1}\\
-c \underbrace{(\sum_{j=1}^l \mat{L}_1^{(j)} \otimes \mat{L}_c^{(j)})(\mat{A}_1 \otimes \mat{E})(\mat{I} \otimes \mat{F})}_{\mat{X}_1\mat{Y}_2}\\
-c \underbrace{(\sum_{j=1}^l (\mat{L}_1^{(j)} \otimes \mat{e}^{(j)})(\mat{I} \otimes \mat{f}^{(j)}))(\mat{A}_1 \otimes \mat{E})(\mat{I} \otimes \mat{F})}_{\mat{X}_2\mat{Y}_2} )^{-1}\\
(\sum_{j=1}^l \mat{L}_1^{(j)} \otimes \mat{L}_2^{(j)}) \mat{x}
\end{array}
\vspace{-1pt}
\end{equation}

\normalsize
Each $\mat{L}_1^{(j)} \otimes \mat{e}^{(j)}$ is a matrix of size $t^2$ by $t$ and $\mat{I} \otimes \mat{f}^{(j)}$ is a matrix of size $t$ by $t^2$. We denote the matrix created by concatenating all  $\mat{L}_1^{(j)} \otimes \mat{e}^{(j)}$ horizontally as $\mat{P}$, \emph{i.e.}, $\mat{P}=[\mat{L}_1^{(1)} \otimes \mat{e}^{(1)},\ldots
, \mat{L}_1^{(l)} \otimes \mat{e}^{(l)}]$; denote the matrix created by stacking all  $\mat{I} \otimes \mat{f}^{(j)}$ vertically as $\mat{Q}$, \emph{i.e.}, $\mat{Q}=[\mat{I} \otimes \mat{f}^{(1)};\ldots;\mat{I} \otimes \mat{f}^{(l)}]$. Obviously, $(\sum_{j=1}^{l}  (\mat{L}_1^{(j)}\otimes \mat{e}^{(j)})(\mat{I}\otimes \mat{f}^{(j)}))$ is equal to $\mat{P}\mat{Q}$. We denote $(\sum_{j=1}^{l} \mat{L}_1^{(j)} \otimes \mat{L}_c^{(j)})(\mat{A}_1 \otimes \mat{E})$ by $\mat{X}_1$; denote  $(\sum_{j=1}^{l} (\mat{L}_1^{(j)}\otimes \mat{e}^{(j)})(\mat{I}\otimes \mat{f}^{(j)})) (\mat{A}_1 \otimes \mat{E})$ by $\mat{X}_2$; denote $\mat{Q}(\mat{A}_1 \otimes \mat{A}_c)$ by $\mat{Y}_1$ and denote $(\mat{I}\otimes \mat{F})$ by $\mat{Y}_2$. Let $\mat{X}$ be $[\mat{P},\mat{X}_1,\mat{X}_2]$ and $\mat{Y}$ be $[\mat{Y}_1; \mat{Y}_2;\mat{Y}_2]$.

With these additional notations, we can rewrite Eq.~\eqref{eq:faxtexact:1} as
\small
\vspace{-5pt}
\begin{equation}\label{eq:faxtexact:2}
\begin{array}{l}
\mathrm{Ker}(\mat{G}(\mathcal{T}), \mat{G}(\mathcal{T}_{p \rightarrow q}))= \mat{y}' (\mat{Z} -c \mat{XY})^{-1}(\sum_{j=1}^l \mat{L}_1^{(j)} \otimes \mat{L}_2^{(j)}) \mat{x}\\
=\mat{y}'(\mat{Z}^{-1}+c \mat{Z}^{-1}\mat{X}(\mat{I}-c\mat{Y}\mat{Z}^{-1}\mat{X})^{-1}\mat{Y}\mat{Z}^{-1})\\
((\sum_{j=1}^{l} \mat{L}_1^{(j)}\otimes \mat{L}_c^{(j)})\mat{x}+(\sum_{j=1}^{l}(\mat{L}_1^{(j)}\otimes \mat{e}^{(j)})(\mat{I} \otimes \mat{f}^{(j)}))\mat{x})\\
\end{array}
\end{equation}
\normalsize
\noindent where the second equation is due to the matrix inverse lemma~\cite{golub1996matrix}. 

{\em Remarks}. In Eq.~\eqref{eq:faxtexact:2}, $\mat{Z}=\mat{I} - c(\sum_{j=1}^{l} \mat{L}_1^{(j)} \otimes \mat{L}_c^{(j)})(\mat{A}_1 \otimes \mat{A}_c)$ does not depend on the candiate $q$. Thus, if we pre-compute its inverse $\mat{Z}^{-1}$, we only need to update $\mat{X}(\mat{I}-c\mat{Y}\mat{Z}^{-1}\mat{X})^{-1}\mat{Y}$ and $\mat{P}\mat{Q}\mat{x}$ for every new candidate. Notice that compared with the original graph kernel (the first equation in Eq.~\eqref{eq:faxtexact:1}), $(\mat{I}-c\mat{Y}\mat{Z}^{-1}\mat{X})$ is a much smaller matrix of $(l+4)t \times (l+4)t$. In this way, we can accelerate the process of computing its inverse without losing the accuracy of graph kernel.

\subsection{Speedup Graph Kernel - Approx Approach}


Note that the graph kernel by Eq.~\eqref{eq:faxtexact:2} is exactly the same as the original method by the first equation in Eq.~\eqref{eq:faxtexact:1}. If we allow some approximation error, we can further speed-up the computation.

{Note that $\mat{A}_c$ is symmetric and its rank-{\em r} approximation can be written as $\hat{\mat{A}}_c = \mat{U} \mat{V}$, where $\mat{U}$ is a matrix of size $t$ by $r$ and \mat{V} is a matrix of size $r$ by $t$.  $\mat{A}_1$ can be approximated as $\hat{\mat{A}}_1=\hat{\mat{A}}_c + \mat{A}_{d1} = \mat{U} \mat{V} + \mat{E}_1\mat{F}_1 = \mat{X}_1\mat{Y}_1$, where $\mat{X}_1=[\mat{U},\mat{E}_1], \mat{Y}_1 = [\mat{V}; \mat{F}_1], \mat{E}_1=[\mat{w}_1, \mat{s}], \mat{F}_1=[\mat{s}';\mat{w}'_1]$, $\mat{s}$ is a zero vector of length $t$ except that the last element is 1, and $\mat{w}_1$ is the weight vector from $p$ to the members in $\mathcal{T}$. Similarly, after $p$ is replaced by a candidate $q$, the weight matrix of the new team can be approximated as $\hat{\mat{A}}_2 = \mat{X}_2\mat{Y}_2$ where $\mat{X}_2=[\mat{U}, \mat{E}_2], \mat{Y}_2 = [\mat{V};\mat{F}_2], \mat{E}_2 = [\mat{w}_2, \mat{s}], \mat{F}_2=[\mat{s}';\mat{w}'_2]$ and $\mat{w}_2$ is the weight vector from $q$ to the members in the new team. }The approximated graph kernel for labeled graphs can be computed as:
\small
\vspace{-5pt}
\begin{equation}\label{eq:fastapprox}
\begin{array}{l}
\hat{\mathrm{Ker}}(\mat{G}(\mathcal{T}), \mat{G}(\mathcal{T}_{p \rightarrow q}) ) =\mat{y}^T (\mat{I} - c \mat{L}_{\times} (\hat{\mat{A}}_1' \otimes \hat{\mat{A}}_2'))^{-1} \mat{L}_{\times} \mat{x}\\
=\mat{y}'(\mat{I} - c \mat{L}_{\times} (\mat{X_1Y_1}) \otimes (\mat{X_2Y_2})) ^{-1} \mat{L}_{\times} \mat{x}\\
=\mat{y}'(\mat{I} -c \mat{L}_{\times} (\mat{X_1}\otimes \mat{X_2})(\mat{Y}_1 \otimes \mat{Y}_2))^{-1} \mat{L}_{\times}\mat{x}\\
=\mat{y}'(\mat{I} +c\mat{L}_{\times}(\mat{X_1}\otimes \mat{X_2})\mat{M} (\mat{Y}_1 \otimes \mat{Y}_2) )\mat{L}_{\times}\mat{x}\\
=\mat{y}' \mat{L}_{\times} \mat{x} + c \mat{y}'(\sum_{j=1}^{l} \mat{L}_1^{(j)} \mat{X_1} \otimes \mat{L}_2^{(j)} \mat{X_2})\mat{M}(\mat{Y}_1 \otimes \mat{Y}_2)\mat{L}_{\times}\mat{x}\\
=(\sum_{j=1}^l (\mat{y}'_1 \mat{L}_1^{(j)} \mat{x}_1)(\mat{y}'_2 \mat{L}_2^{(j)}\mat{x}_2)) +c\\
(\sum_{j=1}^l \mat{y}'_1 \mat{L}_1^{(j)}\mat{X_1} \otimes \mat{y}'_2 \mat{L}_2^{(j)}\mat{X_2} )\mat{M}(\sum_{j=1}^l \mat{Y}_1\mat{L}_1^{(j)}\mat{x}_1 \otimes \mat{Y}_2\mat{L}_2^{(j)}\mat{x}_2)\\
\end{array}
\end{equation}
\normalsize
\noindent where $\mat{M} = (\mat{I} - c(\sum_{j=1}^l \mat{Y}_1 \mat{L}_1^{(j)} \mat{X_1} \otimes \mat{Y}_2 \mat{L}_2^{(j)} \mat{X_2} ))^{-1}$, the second equation is due to the kronecker product property; the third equation is again due to the matrix inverse lemma, the fourth equation is by matrix multiplication distributivity and the last equation is due to the kronecker product property.

{\em Remarks}. The computation of $\mat M$ is much cheaper than the original graph kernel since it is a matrix inverse of size $(r+2)^2\times (r+2)^2$. It was first proposed in~\cite{kang:fastwalk} to explore the low-rank structure of the input graphs to speed-up graph kernel computations. However, in the context of \teamrep, we would need to estimate the low-rank approximation {\em many} times ($O(\sum_{i\in \mathcal{T}/p}d_i)$) when we directly apply the method in \cite{kang:fastwalk}. In contrast,  we only need to compute top-$r$ approximation {\em once} by Eq.~\eqref{eq:fastapprox}. As our complexity analysis (subsection~\ref{sec:alg}) and experimental evaluations (subsection~\ref{sec:exp_efficiency}) show, this brings a few times additional speed-up. \hh{fill in xxx later}

\subsection{Putting Everything Together}\label{sec:alg}

Putting everthing together, we are ready to present our algorithms for \teamrep. Depending on the specific methods for computing the individual graph kernels, we propose two variants.

\subsubsection{Variant \#1: \teamrepalgexact}

We first present our algorithm using the exact graph kernel computation in Eq.~\eqref{eq:faxtexact:2}. The algorithm (\teamrepalgexact) is summarized in Algorithm~\ref{alg:fastexact}. We only need to pre-compute and store $\mat{Z}^{-1}$, $\mat{R}$, $\mat{b}$ and $\mat{l}$ for later use to compute each candidate's score (step 2 and 3). In the loop, the key step is to update $\mat{M}$ involving matrix inverse of size $(l+4)t \times (l+4)t$ which is relatively cheaper to compute (step 17). \hh{the description is good. also add the refernece to step numbers}

\vspace{-8pt}
\begin{algorithm}[!htb]
\caption{\teamrepalgexact}\label{alg:fastexact}

\KwIn{(1) The entire social network $\mat{G}:=\{\mat{A}, \mat{L}\}$, (2)
 original team members $\mathcal{T}$, (3)  person $p$ who  will leave the team, (4)
 starting and ending probability $\mat{x}$ and $\mat{y}$(be uniform by default), and (5) an integer $k$ (the budget)}
\KwOut{Top $k$ candidates to replace person $p$}
\BlankLine
Initialize $\mat{A}_c, \mat{L}_1^{(j)}, \mat{L}_2^{(j)}, j=1,\ldots, l$ \;
Pre-compute  \small $\mat{Z}^{-1} \leftarrow (\mat{I} - c(\sum_{j=1}^{l} \mat{L}_1^{(j)} \otimes \mat{L}_c^{(j)})(\mat{A}_1 \otimes \mat{A}_c) )^{-1}$\;
\normalsize Set \small $\mat{R} \leftarrow (\sum_{j=1}^{l} \mat{L}_1^{(j)}\otimes \mat{L}_c^{(j)})\mat{x}$; $\mat{b}\leftarrow \mat{y}^T \mat{Z}^{-1} \mat{R}$;
$\mat{l}\leftarrow c\mat{y}^T\mat{Z}^{-1}$\;
\normalsize
\For{each candidate $q$ in $\mat{G}$ after pruning}{
	Initialize $\mat{s}\leftarrow$ a zero vector of length $t$ except the last element is 1\;
	Initialize $\mat{w}\leftarrow$ weight vector from $q$ to the new team members\;
	Set $\mat{E}\leftarrow [\mat{w},\mat{s}];\mat{F}\leftarrow [\mat{s}';\mat{w}']$ \;
	Set $\mat{e}^{(j)} \leftarrow$ a $t$ by 1 zero vector except the last element is 1, for $j=1,\ldots,d_n$ \;
	Set $\mat{f}^{(j)} \leftarrow$ a $1 \times t$ zero vector except the last element which is label $j$ assignment for $q$\;
	Set $\mat{P} \leftarrow [\mat{L}_1^{(1)} \otimes \mat{e}^{(1)},\ldots, \mat{L}_1^{(l)} \otimes \mat{e}^{(l)}]$\;
	Set $\mat{Q} \leftarrow [\mat{I} \otimes \mat{f}^{(1)};\ldots;\mat{I} \otimes \mat{f}^{(l)}]$\;
	Compute $\mat{X}_1 \leftarrow (\sum_{j=1}^{l} \mat{L}_1^{(j)}\mat{A}_1 \otimes \mat{L}_c^{(j)}\mat{E})$\;
	Compute $\mat{X}_2 \leftarrow (\sum_{j=1}^{l} \mat{L}_1^{(j)}\mat{A}_1 \otimes \mat{e}^{(j)}\mat{f}^{(j)}\mat{E})$\;
	Compute $\mat{Y}_1 \leftarrow \mat{Q}(\mat{A}_1 \otimes \mat{A}_c)$\;
	Compute $\mat{Y}_2 \leftarrow (\mat{I}\otimes \mat{F})$\;
	Set $\mat{X} \leftarrow [\mat{P},\mat{X}_1,\mat{X}_2], \mat{Y} \leftarrow [\mat{Y}_1; \mat{Y}_2;\mat{Y}_2]$\;
	Update $\mat{M} \leftarrow (\mat{I}-c\mat{Y}\mat{Z}^{-1}\mat{X})^{-1}$\;
	Compute $\mat{r}' \leftarrow  \mat{Z}^{-1} \mat{P}\mat{Q} \mat{x}$\;
	Compute  $\textrm{score}(q)=\mat{b} + \mat{y}^T \mat{r}' + \mat{l}\mat{XMY}(\mat{Z}^{-1}\mat{R}+\mat{r}')$ \;
}
{\bf Return} the top $k$ candidates with the highest scores.
\end{algorithm}
\vspace{-8pt}

The effectiveness and efficiency of \teamrepalgexact\ are summarized in Lemma~\ref{lm:exact:effec} and Lemma~\ref{lm:exact:time}, respectively. Compared with  \teamrepalgbasic, Algorithm~\ref{alg:fastexact} is much faster without losing any recommendation accuracy.

\begin{lemma}\label{lm:exact:effec}{Accuracy of \teamrepalgexact}.
Algorithm~\ref{alg:fastexact} outputs the same set of candidates as \teamrepalgbasic.
\end{lemma}
\begin{proof}
(Sketch) First, according to Lemma~\ref{lm:pruning}, we will not miss a promising candidate during the pruning stage. Second, for each candidate after pruning, Algorithm~\ref{alg:fastexact} calculates its graph kernel exactly the same as Eq.~\eqref{eq:faxtexact:2}, which is in turn the same as Eq.~\eqref{eq:faxtexact:1} and hence Eq.~\eqref{eq:original}. Therefore, after ranking the scores, Algorithm~\ref{alg:fastexact} outputs the same set of candidates as \teamrepalgbasic.
\end{proof}

\begin{lemma}\label{lm:exact:time}{Time Complexity of \teamrepalgexact}
Algorithm~\ref{alg:fastexact} takes $O((\sum_{i\in \mathcal{T}/p}d_i)(lt^5 + l^3t^3))$ in time.
\end{lemma}
\begin{proof}
 (Sketch) After pruning, the number of potential candidates (the number of loops in Algorithm~\ref{alg:fastexact} ) is $O(\sum_{i \in \mathcal{T}/p} d_i)$. In every loop, computing $\mat{X}_1$,$\mat{X}_2$ and $\mat{Y}_1$ take $O(lt^5)$; computing $\mat{M}$ takes $O(lt^5+l^3t^3)$ and computing the score($q$) takes $O(lt^3)$. Putting everything together, the time complexity of Algorithm~\ref{alg:fastexact} is $O((\sum_{i\in \mathcal{T}/p}d_i)(lt^5 + l^3t^3))$.
\end{proof}

\subsubsection{Variant \#2: \teamrepalgapp}

By using Eq.~\eqref{eq:fastapprox} to compute the graph kernel instead, we propose an even faster algorithm (\teamrepalgapp), which is summarized in Algorithm~\ref{alg:fastappro}. In the algorithm, we only need to compute the top $r$ eigen-decomposition for $\mat{A}_c$ once (step 2), and use that to update the low rank approximation for every new team. Besides, when we update $\mat{M}$, a matrix inverse of size $(r+2)^2 \times (r+2)^2$ (step 14), the time is independent of the team size.
\vspace{-8pt}
\begin{algorithm}[!htb]
\caption{\teamrepalgapp}\label{alg:fastappro}

\KwIn{(1) The entire social network $\mat{G}:=\{\mat{A}, \mat{L}\}$, (2)
 original team members $\mathcal{T}$, (3)  person $p$ who  will leave the team, (4)
 starting and ending probability $\mat{x}$ and $\mat{y}$.(be uniform by default), and (5) an integer $k$ (the budget)}
\KwOut{Top $k$ candidates to replace person $p$}
\BlankLine
Initialize $\mat{A}_c, \mat{L}_1^{(j)}, \mat{L}_2^{(j)}, j=1,\ldots, l$ \;
Compute top $r$ eigen-decomposition for $\mat{A}_c$: $\mat{U} \mat{\Lambda} \mat{U}' \leftarrow \mat{A}_c$ \;
Set $\mat{V} \leftarrow \mat{\Lambda} \mat{U}'$\;
Initialize $\mat{s} \leftarrow$ a zero vector of length $t$ except the last element is 1\;
Initialize $\mat{w}_1 \leftarrow $ weight vector from $p$ to $\mathcal{T}$\;
Set $\mat{E}_1\leftarrow [\mat{w}_1,\mat{s}], \mat{F}_1\leftarrow [\mat{s}';\mat{w}'_1]$ \;
Set $\mat{X}_1 \leftarrow [\mat{U},\mat{E}_1]$ , $\mat{Y}_1 \leftarrow [\mat{V}; \mat{F}_1]$\;
\For{each candidate $q$ in $\mat{G}$ after pruning}{
	Initialize $\mat{w}_2 \leftarrow$ weight vector from $q$ to the new team members \;
	Set $\mat{E}_2\leftarrow [\mat{w}_2,\mat{s}], \mat{F}_2\leftarrow [\mat{s}';\mat{w}'_2]$ \;
	Set $\mat{X}_2 \leftarrow [\mat{U},\mat{E}_2]$ , $\mat{Y}_2 \leftarrow [\mat{V}; \mat{F}_2]$\;
	Compute $\mat{S} \leftarrow \sum_{j=1}^{l} \mat{y}'_1 \mat{L}_1^{(j)} \mat{X}_1 \otimes \mat{y}'_2 \mat{L}_2^{(j)} \mat{X}_2$\;
	Compute \small $\mat{T} \leftarrow \sum_{j=1}^{l} \mat{Y}_1\mat{L}_1^{(j)} \mat{x}_1 \otimes \mat{Y}_2 \mat{L}_2^{(j)} \mat{x}_2)$\normalsize\;
	Update \small $\mat{M} \leftarrow (\mat{I} - c(\sum_{j=1}^l \mat{Y}_1 \mat{L}_1^{j} \mat{X}_1 \otimes \mat{Y}_2 \mat{L}_2^{j} \mat{X}_2 ))^{-1}$\normalsize\;
	Set  \small $\textrm{score}(q)=(\sum_{j=1}^l (\mat{y}'_1 \mat{L}_1^{(j)} \mat{x}_1)(\mat{y}'_2 \mat{L}_2^{(j)}\mat{x}_2))  + c\mat{SMT}$ \normalsize\;
}
{\bf Return} the top $k$ candidates with the highest scores.
\end{algorithm}
\vspace{-8pt}

The effectiveness and efficiency of \teamrepalgapp\ are summarized in Lemma~\ref{lm:approx:effec} and Lemma~\ref{lm:approx:time}, respectively. Compared with  \teamrepalgbasic\ and \teamrepalgexact, Algorithm~\ref{alg:fastappro} is even faster; and the only place it introduces the approximation error is the low-rank approximation of $\mat{A}_c$ (step 2). 

\begin{lemma}\label{lm:approx:effec}{Accuracy of \teamrepalgapp}.
If $\mat{A}_c=\mat{U}\mat{\Lambda}\mat{U}'$ holds, Algorithm~\ref{alg:fastappro} outputs the same set of candidates as \teamrepalgbasic.
\end{lemma}
\begin{proof} Omitted for brevity.
\end{proof}

\begin{lemma}\label{lm:approx:time}{Time Complexity of \teamrepalgapp}
Algorithm~\ref{alg:fastappro} takes $O((\sum_{i \in \mathcal{T}/p} d_i)(lt^2r +r^6))$ in time.
\end{lemma}
\begin{proof} Omitted for brevity.
\end{proof}

%% file: 006exp.tex
%

In this section, we present the experimental evaluations. The experiments are designed to answer the following questions:
\vspace{-5pt}
\begin{itemize}
\setlength{\itemsep}{0.1\baselineskip}
\item {\em Effectiveness}: How accurate are the proposed algorithms for \teamrep?
\item {\em Efficiency:} How scalable are the proposed algorithms?
\end{itemize}

\subsection{Datasets}
{\em DBLP}. DBLP dataset\footnote{http://arnetminer.org/citation} provides bibliographic information on major computer science journals and proceedings. We use it to build a co-authorship network where each node is an author and the weight of each edge stands for the number of papers the two corresponding authors have co-authored. The network constructed has $n=916,978$ nodes and $m=3,063,244$ edges. We use the conferences (\emph{e.g.,} KDD, SIGMOD, CVPR, etc) as the skill of the authors. \hh{liangyue: if we only have 41 conferences, how can they cover all 900K authors? one way to get around this is not to mention l in table 2} For a given paper, we treat all of its co-authors as a team. Alternatively, if a set of authors co-organize an event (such as a conference), we also treat them as a team.

{\em Movie}. This dataset\footnote{http://grouplens.org/datasets/hetrec-2011/} is an extension of MovieLens dataset, which links movies from MovieLens with their corresponding IMDb webpage and Rotten Tomatoes review system. It contains information of 10,197 movies, 95,321 actors/actress and 20 movie genres (\emph{e.g.}, action, comedy, horror, etc.). Each movie has on average 22.8 actors/actress and 2.0 genres assignments. We set up the social network of the actors/actresses where each node represents one actor/actress and the weight of each edge is the number of movies the two linking actors/actresses have co-stared. We use the movie genres that a person has played as his/her skills. For a given movie, we treat all of its actors/actress as a team.

{\em NBA}. The NBA dataset\footnote{http://www.databasebasketball.com} contains NBA and ABA statistics from the year of 1946 to the year of 2009. It has information of 3,924 players and 100 teams season by season. We use players' positions as their skill labels, including {\em guard}, {\em forward} and {\em center}. The edge weight of the player network stands for the number of seasons that the two corresponding nodes/individuals played in the same team.

The statistics of these three datasets are summarized in Table~\ref{tab:data}. All the experiments are run on a Windows machine with 16 GB memory and Intel i7-2760QM CPU.

\begin{table}[t]
\caption{Summary of Datasets.}
\vspace{-8pt}
\label{tab:data}\centering
\begin{tabular}{c||c|c|c}
  \hline
  Data & n & m &  \# of teams\\ \hline
  {\em DBLP} & 916,978 & 3,063,244 & 1,572,278\\
  {\em Movie} & 95,321& 3,661,679 &  10,197\\
  {\em NBA} & 3,924 & 126,994 &  1,398\\
  \hline
\end{tabular}
\vspace{-15pt}
\end{table}

{\em Repeatability of Experimental Results.} All the three datasets are publicly available. We will release the code of the proposed algorithms through authors' website.

\subsection{Effectiveness Results}

{\em A. Qualitative Evaluations}. We first present some case studies on the three datasets to gain some intuitions.

{\em Case studies on DBLP}. \input{061dblpcase}

We also consider a bigger team, i.e, the organizing committee of {\em KDD 2013}. After filtering those not in {\em DBLP}, we have 32 people in the committee team. We use their co-authorship network as their social network. Suppose one of the research track co-chairs {\em Inderjit Dhillon} becomes unavailable and we are searching for another researcher who can fill in this critical role in organizing {\em KDD 2013}.  The top five candidates our algorithm recommends are {\em Philip S. Yu}, {\em Jiawei Han}, {\em Christos Faloutsos}, {\em Bing Liu} and {\em Wei Wang}. The results are consistent with the intuitions - all of these recommended researchers are highly qualified - not only have they made remarkable contributions to the data mining field, but also they have strong ties with the remaining organizers of {\em KDD 2013}. For example, {\em Liu} is the current chair of KDD executive committee; {\em Wang} is one of the research track program chairs for {\em KDD 2014}; and {\em Faloutsos} was the PC co-chair of {\em KDD 2003}, etc.

\begin{figure}[t]
\centering
\includegraphics[width=0.45\textwidth, height=45mm]{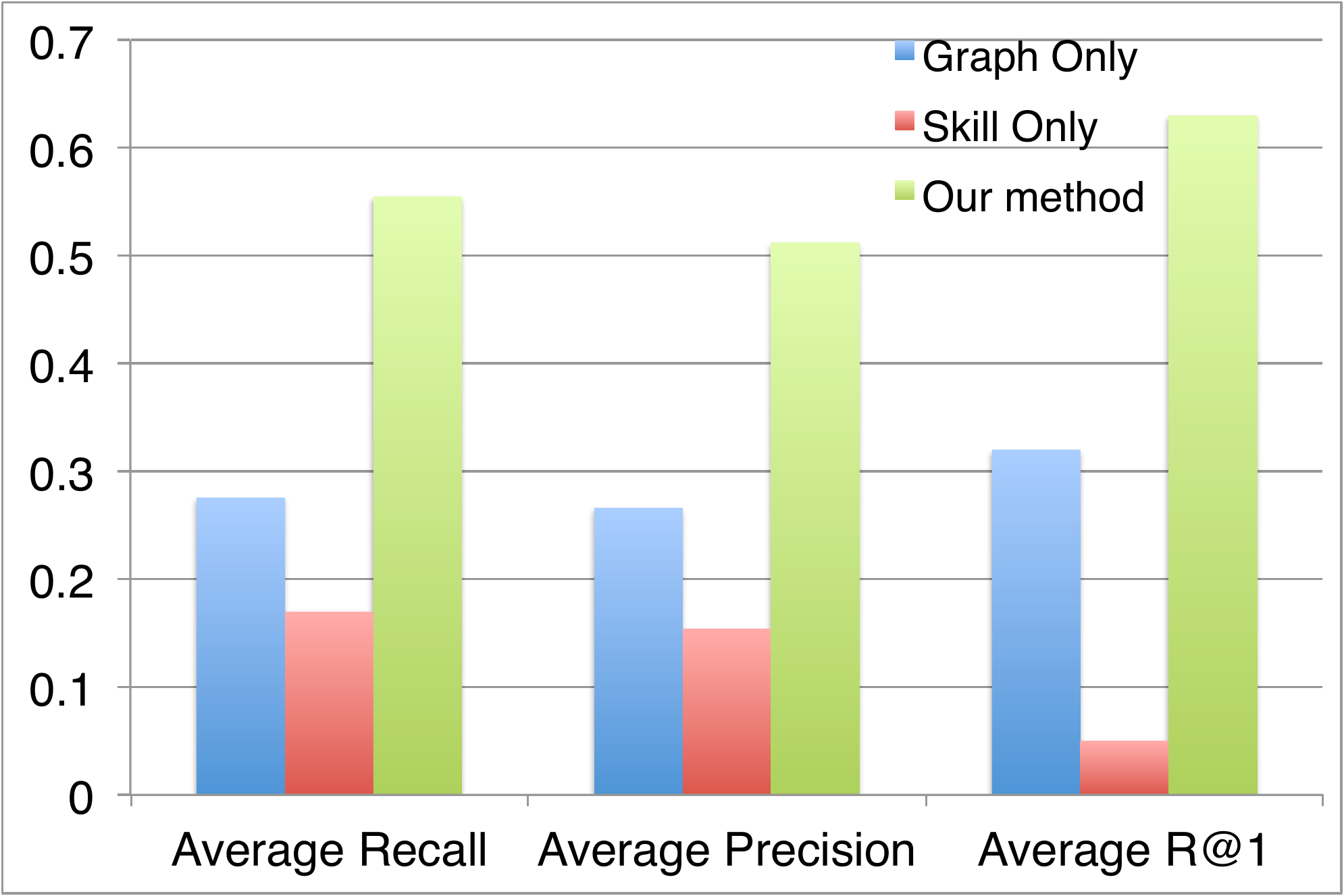}
\vspace{-10pt}
\caption{The average recall, average precision and R$@$1 of the three comparison methods. Higher is better.}
\vspace{-20pt}
\label{fig:precision_recall}
\end{figure}

{\em Case studies on Movie}.
Assuming actor {\em Matt Damon} became unavailable when filming the epic war movie \emph{Saving Private Ryan} (1998) and we need to find an alternative actor who can play \emph{Ryan}'s role in the movie. The top five recommendations our algorithm gives are: {\em Samuel L. Jackson}, {\em Steve Buscemi}, {\em Robert De Niro}, {\em Christopher Walken}, {\em Bruce Willis}. As we know, \emph{Saving Private Ryan} is a movie of \emph{ action} and \emph{drama} genres. Notice that both {\em Damon} and {\em Jackson} have participated in many movies of \emph{drama}, \emph{thriller} and \emph{action} genres, hence {\em Jackson} has the acting skills required to play the role in this movie. Moreover, {\em Jackson} has co-played with {\em Tom Sizemore}, {\em Vin Diesel}, {\em Dale Dye}, {\em Dennis Farina}, {\em Giovanni Ribisi} and {\em Ryan Hurst} in the crew before. The familiarity might increase the harmony of filming the movie with others.

{\em Case studies on NBA}.
Let us assume that {\em Kobe Bryant} in Los Angeles Lakers was hurt during the regular season in 1996 and a bench player is badly wanted. The top five replacements our algorithm recommends are: {\em Rick Fox}, {\em A.c. Green}, {\em Jason Kidd}, {\em Brian Shaw} and {\em Tyronn Lue}. As we know, {\em Bryant} is a guard in NBA. Among the five recommendations, {\em Kidd}, {\em Shaw} and {\em Lue} all play as guards. More importantly, {\em Jason}, {\em Brian} and {\em Tyronn} have played with 9, 7 and 9 of the rest team members on the same team in the same season for multiple times. Therefore, it might be easier for them to maintain the moment and chemistry of the team which is critical to winning the game.

{\em B. Quantitative Evaluations}. Besides the above case studies, we also perform quantitative evaluations to compare the proposed algorithms with the alternative choices, including (a) only with {\em structure matching} and not including $\mat{L}_{\times}$ in Eq.~\eqref{eq:original} (Graph Only), and (b) only with {\em skill matching}  and using cosine similarity of skill vectors as scores (Skill Only).

{\em User studies}. We perform a user study as follows. we choose 10 papers from various fields, replace one author of each paper and run the three comparison methods, and each of them recommends top five candidates. Then, we mix the outputs (15 recommendations in total) and ask users to (a) mark exactly one best replacement; (b) mark all good replacements from the list of 15 recommended candidates. The results are presented in Fig.~\ref{fig:precision_recall}, Fig.~\ref{fig:userstudy_recall} and Fig.~\ref{fig:userstudy_precision}, respectively. As we can see from these figures, the proposed method (the green bar) is best in terms of both precision and recall. For example, the average recalls by our method, by `Graph Only' and by `Skill Only' are 55\%, 28\%, 17\%, respectively. As for different papers, our method wins 9 out-of 10 cases (except for `paper 2' where `Skill Only' is best). 

\begin{figure*}[t]
\centering
\begin{minipage}[b]{0.49\textwidth}
	\includegraphics[width=\textwidth, height=50mm]{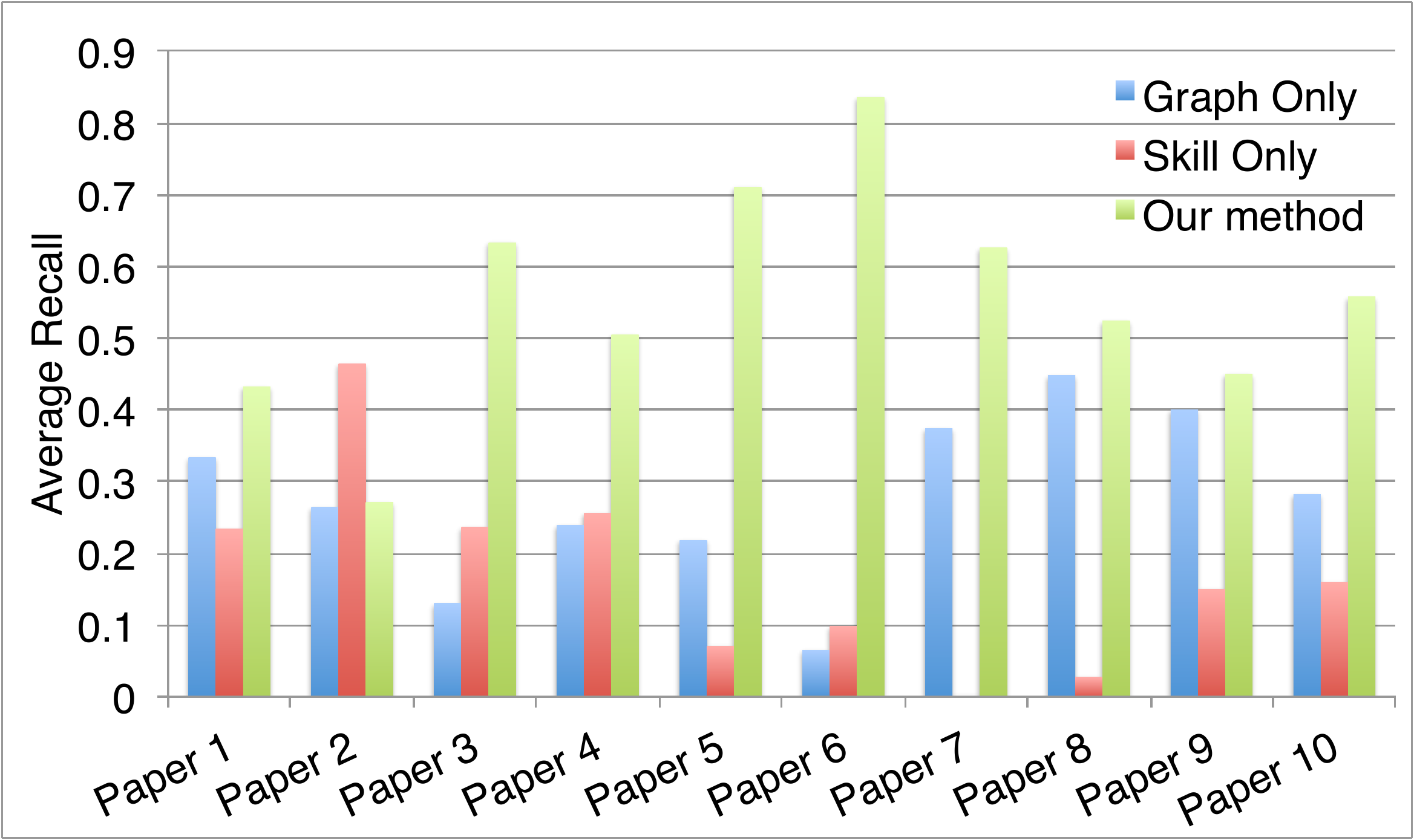}
	\vspace{-10pt}
	\caption{Recall for different papers. Higher is better.}
	\vspace{-10pt}
	\label{fig:userstudy_recall}
\end{minipage}~
\begin{minipage}[b]{0.49\textwidth}
	\includegraphics[width=\textwidth, height=50mm]{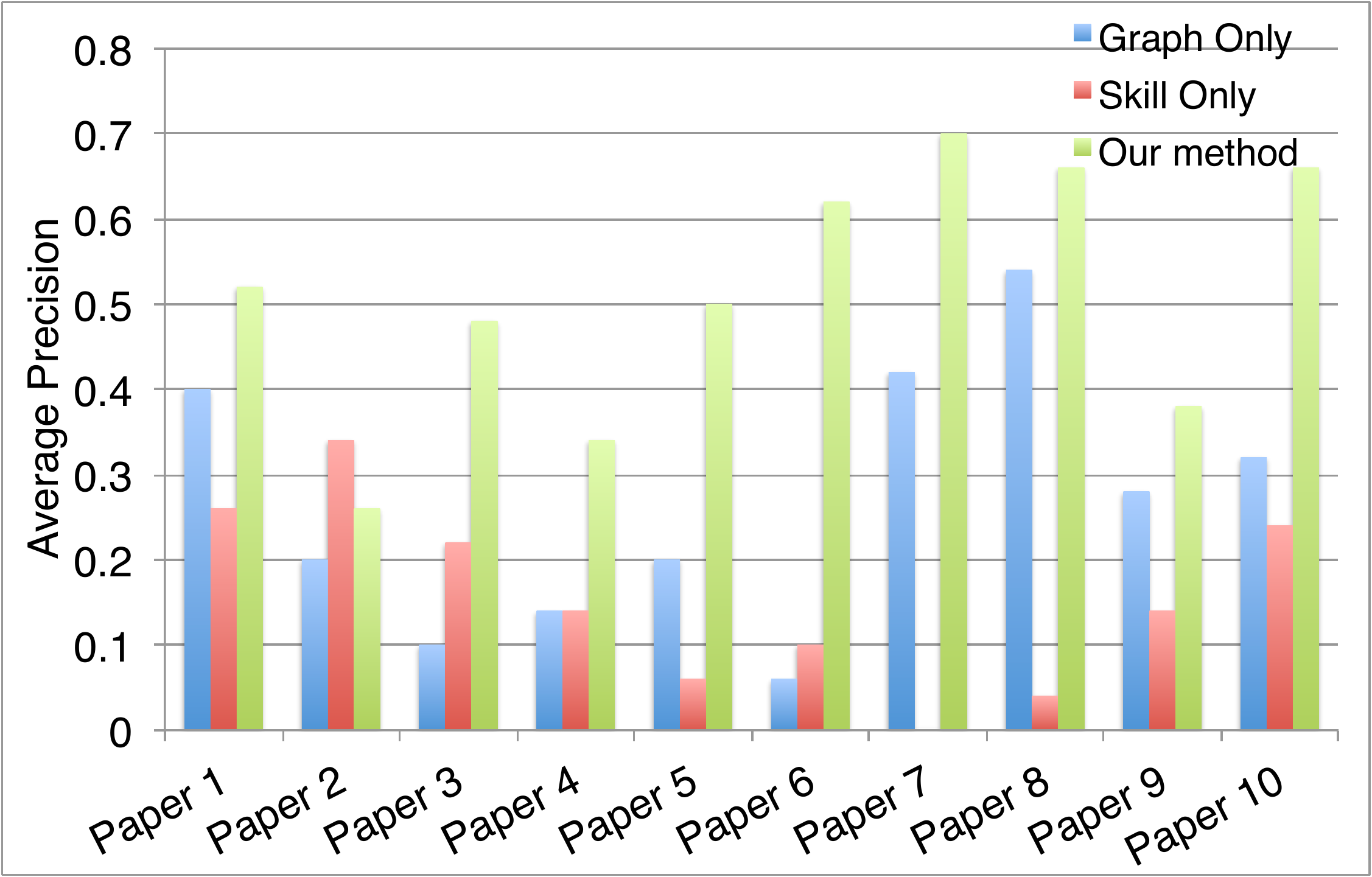}
	\vspace{-10pt}
	\caption{Precision for different papers. Higher is better.}
	\vspace{-10pt}
	\label{fig:userstudy_precision}
\end{minipage}
\end{figure*}

{\em Author alias prediction}. In {\em DBLP}, some researchers might have  multiple name identities/alias. For example, in some papers, {\em Alexander J. Smola} might be listed as {\em Alex J. Smola}, {\em Zhongfei (Mark) Zhang} might be listed as {\em Zhongfei Zhang}, etc. For such an author, we run the team replacement algorithm on those papers s/he was involved to find top-{\em k} replacement. If his/her other alias appears in the top-{\em k} recommended list, we treat it as a {\em hit}. The average accuracy of different methods is shown in Fig.~\ref{fig:accuracy}. Again, our method performs best.

\hide{
Since there is no ground-truth as to which candidates are better, it's hard for us to measure the effectiveness quantitatively. However we can make use of name ambiguity in DBLP, \emph{i.e.}, the same author may have multiple name variations, and replace one name variation in a paper and regard another variation as the ground truth. We choose 6 such authors in DBLP: Chengxiang Zhai vs. ChengXiang Zhai,  HongJiang Zhang vs. Hong-Jiang Zhang, Zhongfei Zhang vs. Zhongfei (Mark) Zhang, Alexander J. Smola vs. Alex J. Smola, David Wai-Lok Cheung vs. David W. Cheung, Tieniu Tan vs. T. N. Tan. We test on 70 papers with these names, replace one name variation and the goal is to expect another one recommended. For each comparison method, if the ground truth appears in the top $k$ recommendation list of the method, we regard it as a {\em hit}. Varying the value of budget $k$, we can compare the average accuracy of different methods shown in Fig.~\ref{fig:accuracy}. Our method achieves the highest accuracy and Label Only's accuracy is always 0 in this case.
}

\begin{figure*}[t]
\centering
\begin{minipage}[b]{0.49\textwidth}
\includegraphics[width=\textwidth, height=50mm]{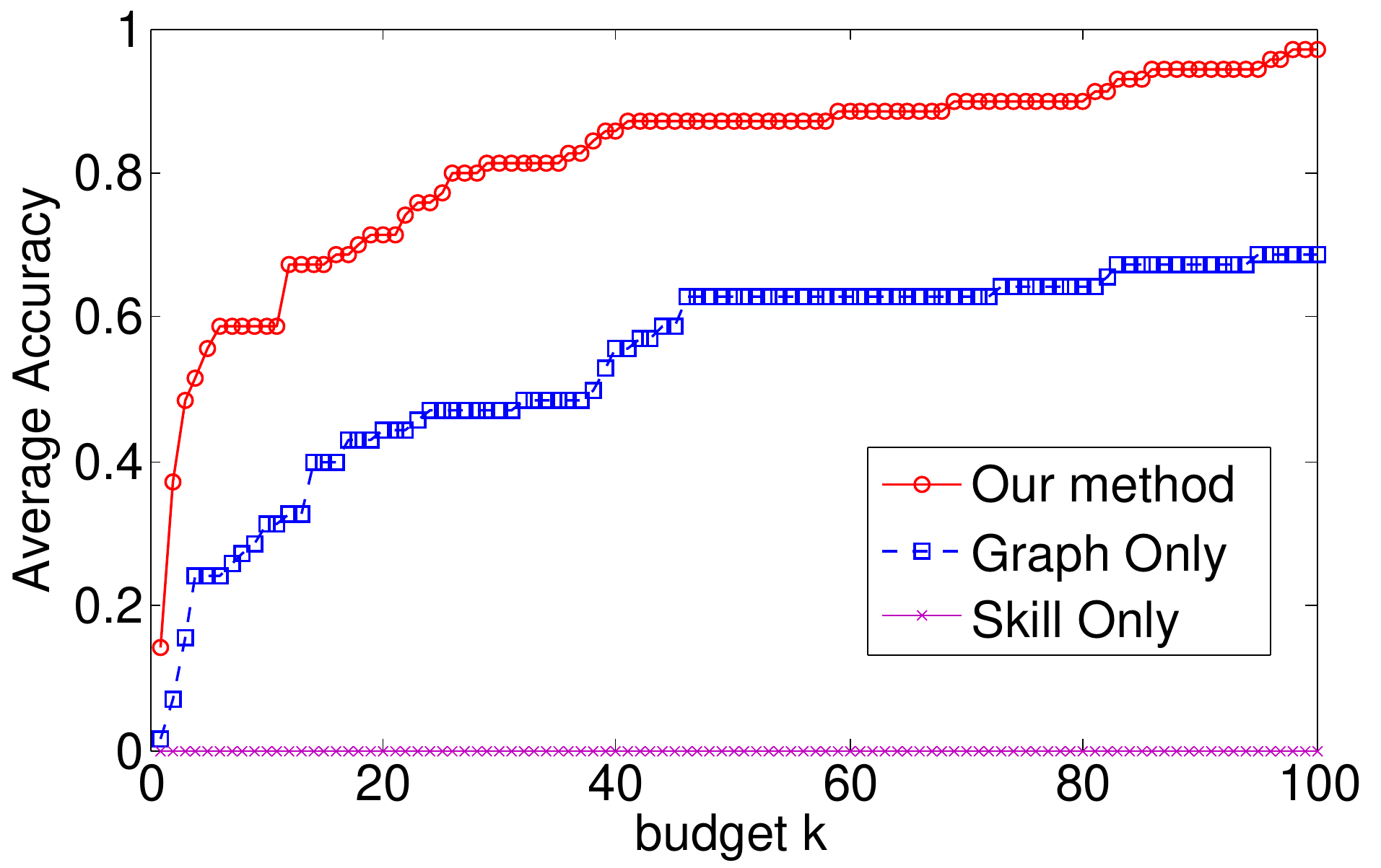}
\vspace{-10pt}
\caption{Average accuracy vs. budget $k$. Higher is better.}
\vspace{-10pt}
\label{fig:accuracy}
\end{minipage}~
\begin{minipage}[b]{0.49\textwidth}
\includegraphics[width=\textwidth, height=47mm]{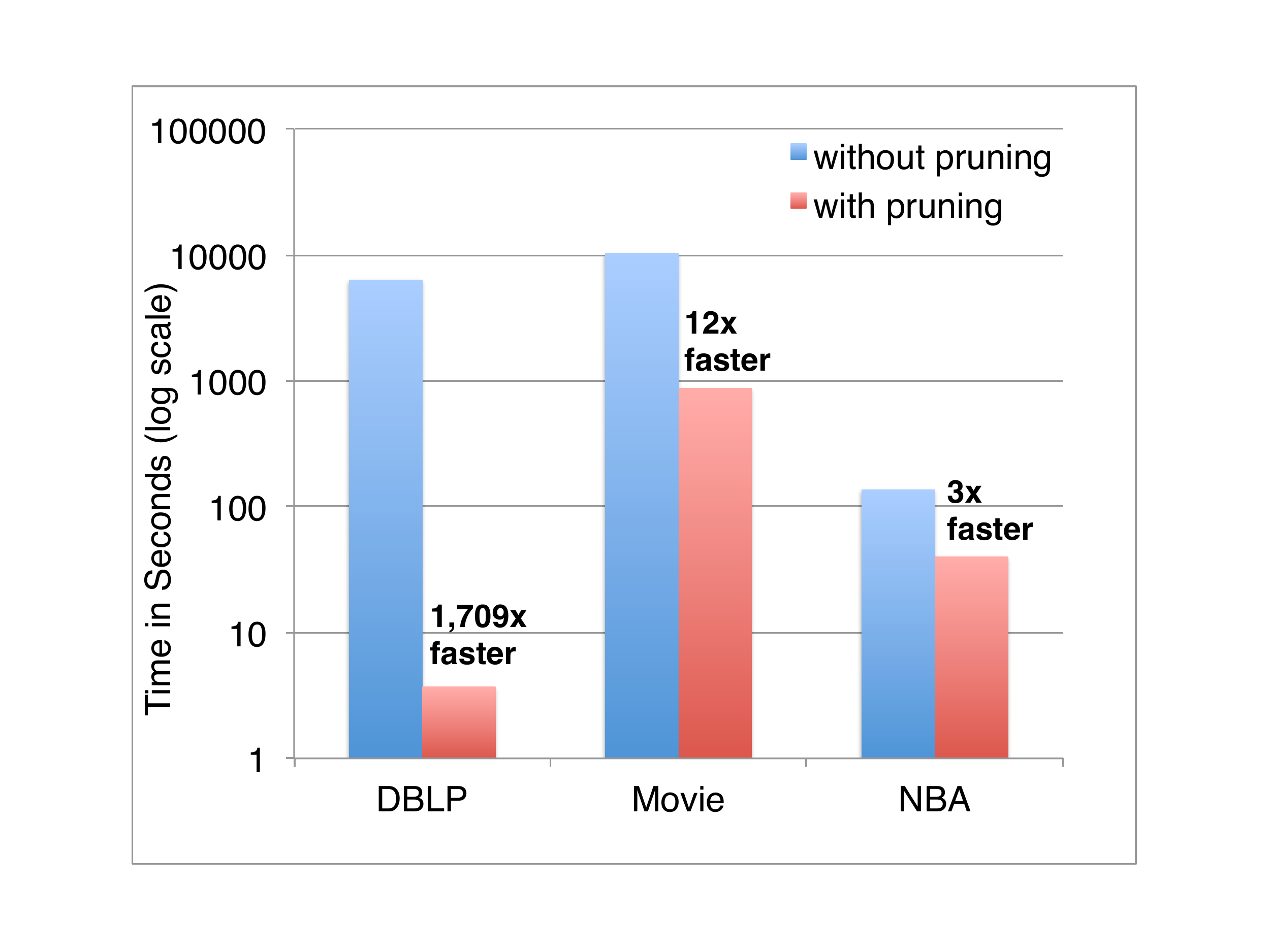}
\vspace{-10pt}
\caption{Time Comparisons before and after pruning on three datasets. Notice time is in log-scale.}
\vspace{-10pt}
\label{fig:time_pruning}
\end{minipage}
\end{figure*}

\subsection{Efficiency Results}\label{sec:exp_efficiency}


{\em A. The speed-up by pruning}. To demonstrate the benefit of our pruning strategy, we run \teamrepalgbasic\ with and without pruning on the three datasets and compare their running time. For {\em DBLP}, we choose the authors of paper~\cite{DBLP:conf/www/LiTZLLH07} (6 authors); for {\em Movie}, we select the film crew of \emph{Titanic} (1997) (22 actors/actresses); for {\em NBA}, we pick the players on Los Angeles Lakers in year 1996 (17 players). The result is presented in Fig.~\ref{fig:time_pruning}. As we can see, the pruning step itself brings significant savings in terms of running time, especially for larger graphs (e.g., {\em DBLP} and {\em Movie}). Notice that according to Lemma~\ref{lm:pruning}, we do not sacrifice any recommendation accuracy by pruning.


{\em B. Further speedup}. Next, we vary the team sizes and compare the running time of \teamrepalgbasic\ with \teamrepalgexact (exact methods); and Ark-L~\cite{kang:fastwalk} with \teamrepalgapp\ (approximate methods). For \teamrepalgbasic\ and Ark-L, we apply the same pruning step as their pre-processing step. The results are presented in Fig.~\ref{fig:time_exactalgs} and Fig.~\ref{fig:time_approxalgs}, respectively. We can see that the proposed \teamrepalgexact\ and \teamrepalgapp\ are much faster than their alternative choices, especially when team size is large. Notice that Ark-L is the best known methods for approximating random walk based graph kernel.

\begin{figure*}[t]
\centering
\begin{minipage}[b]{0.48\textwidth}
\includegraphics[width=\textwidth, height=48mm]{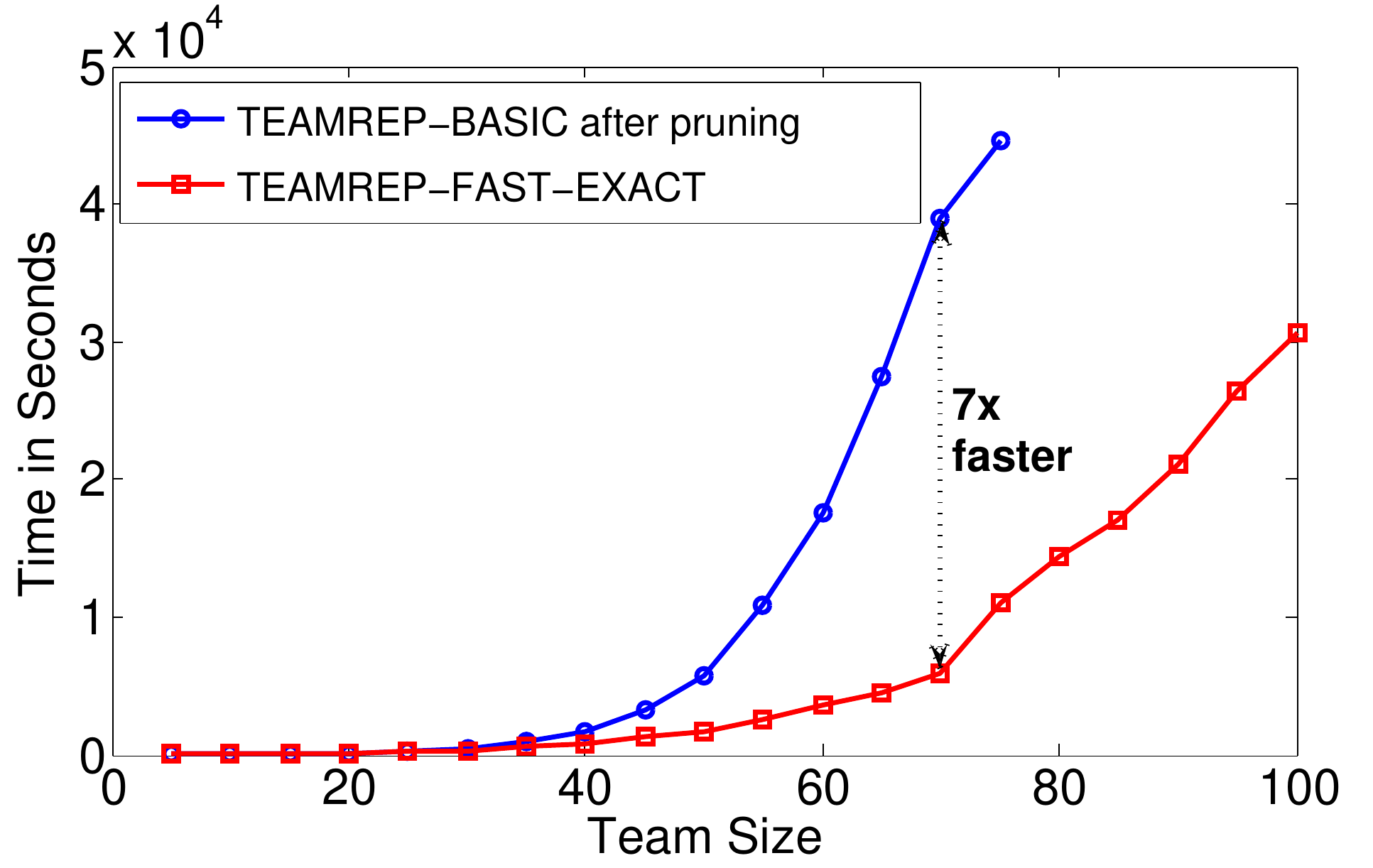}
\vspace{-10pt}
\caption{Time Comparison between \teamrepalgbasic ~ and \teamrepalgexact. \teamrepalgexact ~is on average 3$\times$ faster. \teamrepalgbasic\ takes more than 10 hours when team size = 70. \hh{liangyue: can you move '7x faster' up - between the blue and red curves?}}
\vspace{-15pt}
\label{fig:time_exactalgs}
\end{minipage}~
\begin{minipage}[b]{0.48\textwidth}
\includegraphics[width=\textwidth, height=50mm]{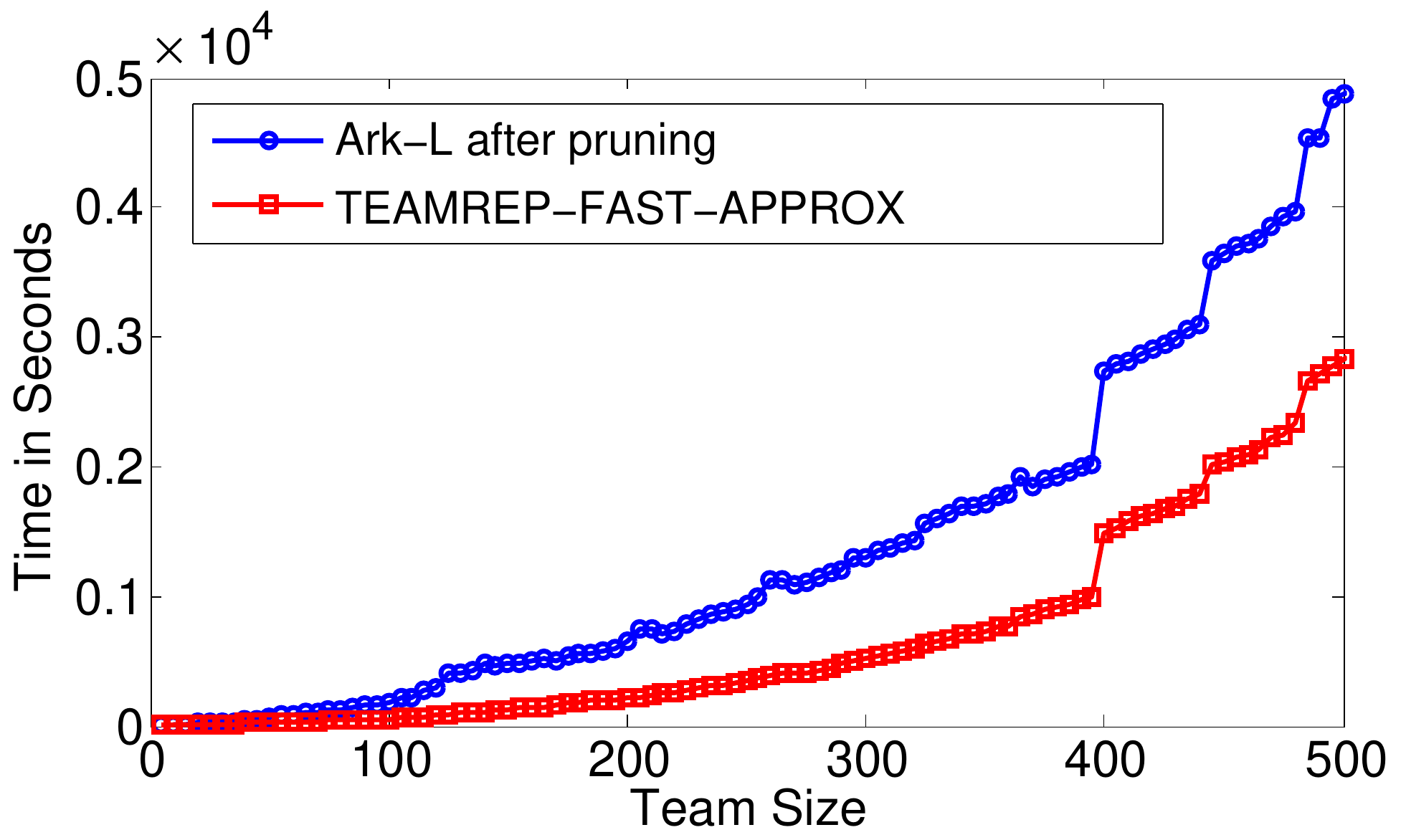}
\vspace{-10pt}
\caption{Time Comparisons between Ark-L[20] and \teamrepalgapp. \teamrepalgapp ~is on average 3$\times$ faster. \hh{for our methods, we do not need to say 'after pruning' in the legend - same for fig 6.}}.
\vspace{-15pt}
\label{fig:time_approxalgs}
\end{minipage}
\end{figure*}

{\em C. Scalability}. \hh{let us call them edges as opposed to links, to be consistent. also change that in the x-axis in fig8-9} To test the scalability of our \teamrepalgexact ~and \teamrepalgapp ~algorithms, we sample a certain percentage of edges from the entire {\em DBLP} network and run the two proposed algorithms on teams with different sizes. The results are presented in Fig.~\ref{fig:scale_fastaexact} and Fig.~\ref{fig:scale_fastapprox}, respectively. As we can seen, both algorithms enjoy a {\em sub-linear} scalability w.r.t. the total number of edges of the input graph ($m$).

\begin{figure*}[t]
\centering
\begin{minipage}[b]{0.49\textwidth}
\includegraphics[width=\textwidth, height=50mm]{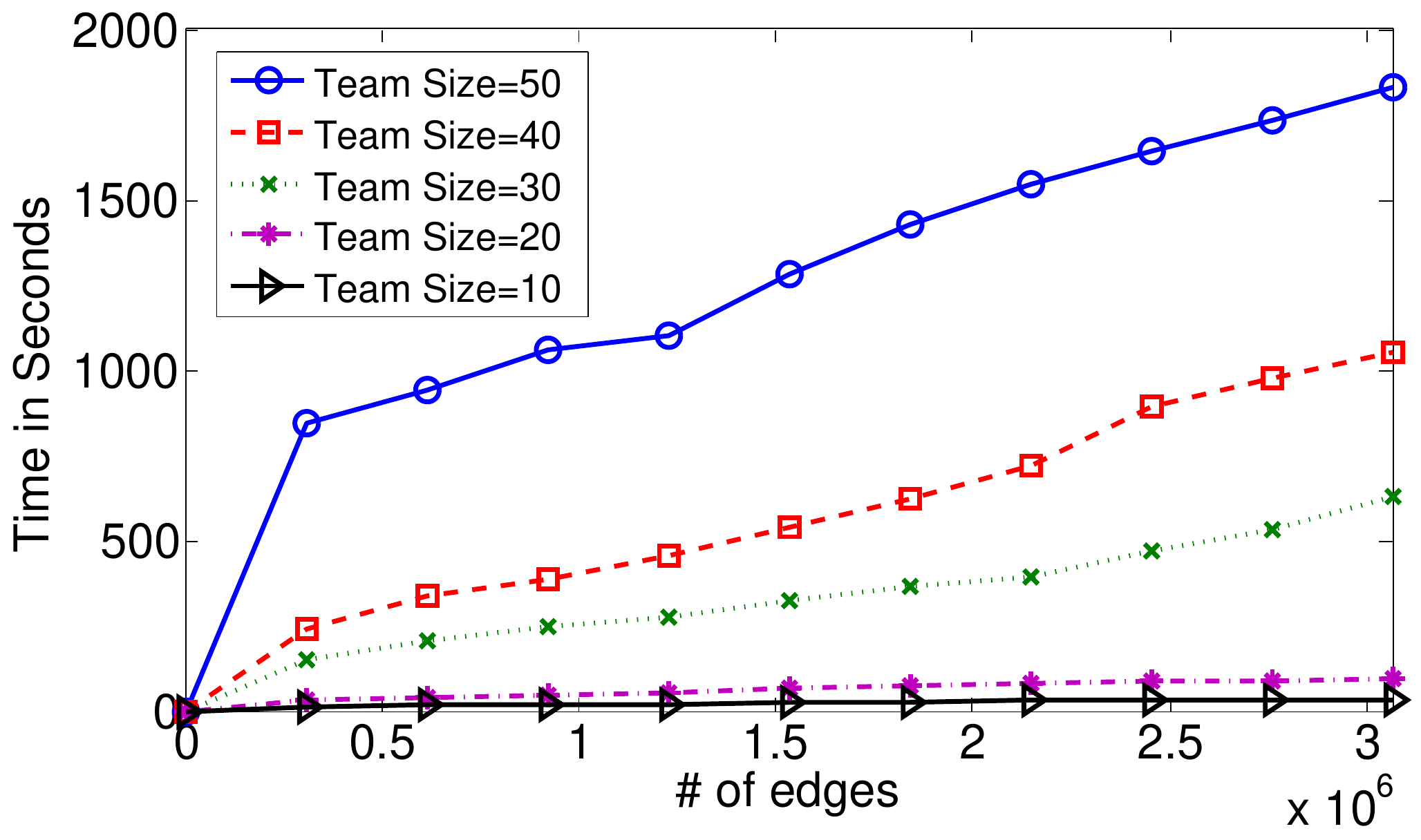}
\vspace{-10pt}
\caption{Running time of \teamrepalgexact\ vs. graph size.  \teamrepalgexact\ scales sub-linearly w.r.t. the number of edges of the input graph. }
\vspace{-15pt}
\label{fig:scale_fastaexact}
\end{minipage}~
\begin{minipage}[b]{0.49\textwidth}
\includegraphics[width=\textwidth, height=50mm]{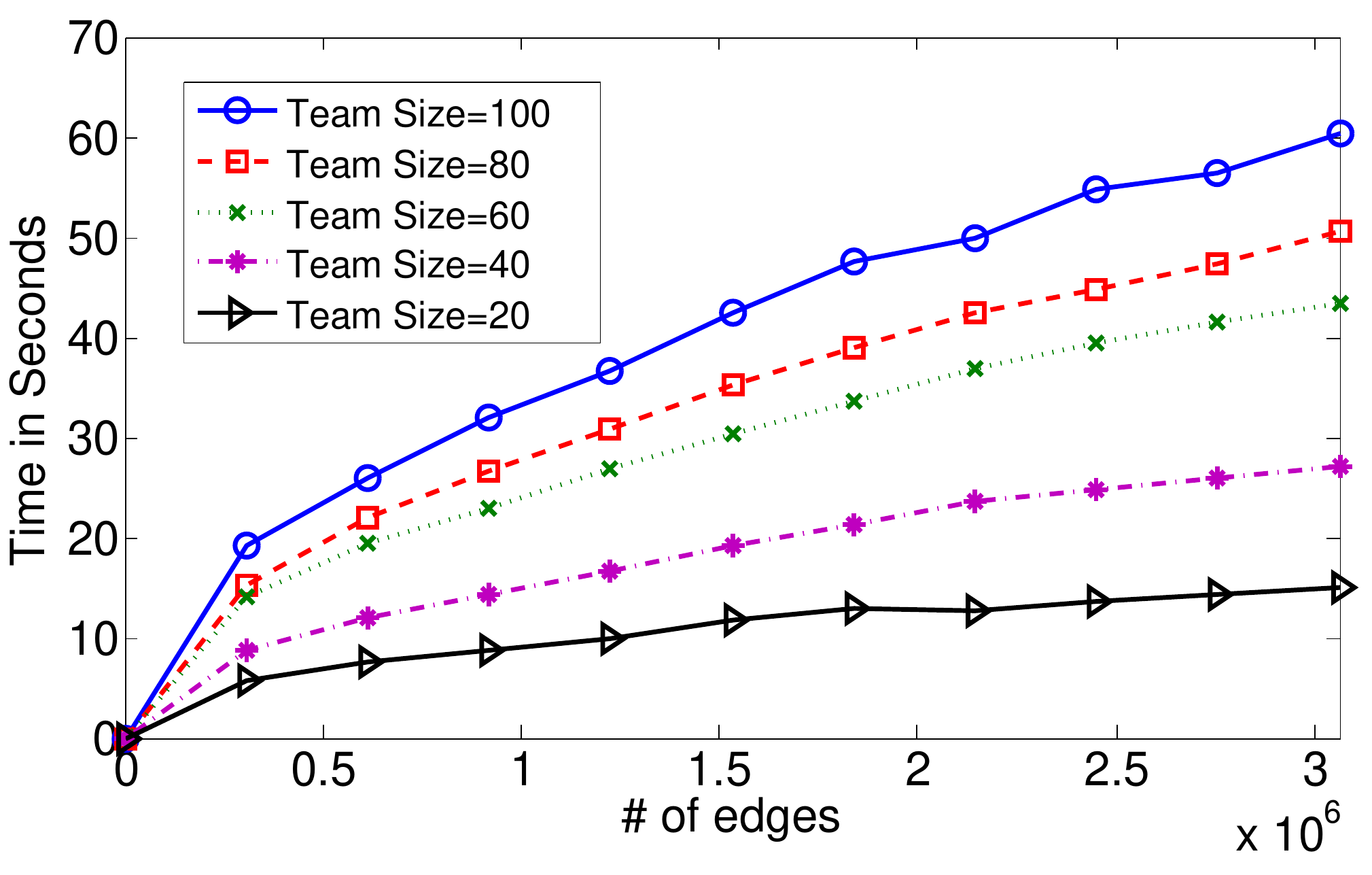}
\vspace{-10pt}
\caption{Running time vs. graph size. \teamrepalgapp\ scales sub-linearly w.r.t. the number of edges of the input graph.}
\vspace{-15pt}
\label{fig:scale_fastapprox}
\end{minipage}
\end{figure*}

%% file: 061dblpcase.tex
Let us take a close look at Fig.~\ref{fig:example}, which shows a screenshot of our current demo system. The original team is shown on the left side and the person leaving the team ({\em Philip S. Yu}) is represented by a node (diagram) with larger radius. If the user clicks a replacement from the recommendation list (on the top), the system will show the new team on the right side. Here, We introduced a novel visualization technique to represent  authors' relationships and their expertise within one unique graph visualization. Particularly, in this visualization, the authors are shown as voronoi diagrams~\cite{du1999centroidal}. The authors' expertise is visualized as the voronoi cells inside the diagram, that is, each cell indicates a type of expertise.  We use different color hues to identify different expertise types and use the color saturations to encode the author's strength in that expertise type. For example, if {\em KDD} is represented in orange, a bright orange cell in a voronoi diagram means the author has a strong expertise in {\em KDD}. In contrast, a white cell indicates the author's lacking of the corresponding expertise. To facilitate visual comparison of different authors, we fix the position of these expertise cells across different diagrams so that, for example, {\em KDD} is always shown at the left side of the author diagrams.  These voronoi diagrams are connected by links indicating the authors' relationships. The strength of the relationship are presented by the line thickness.

Fig.~\ref{fig:example} visualizes the team structures before and  after {\em Philip S. Yu} becomes unavailable in the team writing~\cite{DBLP:journals/tkde/PeiWLWWY06}. Our algorithm's top recommendation is {\em Jiawei Han}. As we can see, both {\em Han} and {\em Yu} posses very similar skills and are renowned for their extraordinary contributions in the data mining and databases community. Moreover, {\em Han} has collaborated with almost each of the rest authors/team members. Looking more closely, we find {\em Han} preserves several key triangle sub-structures in the original team: one with {\em Ke Wang} and {\em Jian Pei}, and the other with {\em Haixun Wang} and {\em Jian Pei}. These triangle sub-structures might play a critical role in accomplishing sub-tasks in writing the paper.

%% file: 002related.tex
In this section, we review the related work in terms of (a) team formation, (b) recommendation and expert finding, and (c) graph kernel.

\textbf{Team Formation.} Team formation studies the problem of assembling a team of people to work on a project. To ensure success, the selected team members should possess the desired skills and have strong team cohesion, which is first studied in ~\cite{DBLP:conf/kdd/LappasLT09}. As follow-up work, Anagnostopoulos et al~\cite{DBLP:conf/www/AnagnostopoulosBCGL12} studies forming teams to accommodate a sequence of tasks arriving in an online fashion and Rangapuram et al~\cite{DBLP:conf/www/RangapuramBH13} allows incorporating many realistic requirements into team formation based on a generalization of the densest subgraph problem. With the presence of the underlying social network, the set cover problem is complicated by the goal of lowering the communication cost at the same time. Bogdanov et al~\cite{DBLP:conf/pkdd/BogdanovBBBS13} studies how to extract a diversified group pulled from  strong cliques given a network; this ensures  that the group is both comprehensive and representative of the whole network. Cummings and Kiesler~\cite{DBLP:conf/cscw/CummingsK08} find that prior working experience is the best predictor of collaborative tie strength. To provide insights into designs of online communities and organizations, the systematic differences in appropriating social softwares among different online enterprise communities is analyzed in \cite{DBLP:conf/chi/MullerEMPRG12}.  The patterns of informal networks and communication in distributed global software teams using social network analysis is also investigated in~\cite{DBLP:conf/cascon/ChangE07}. Specific communication structures are proven critical to new product development delivery performance and quality~\cite{DBLP:conf/chi/CataldoE12}. For assessing the skills of players and teams in online multi-player games and team-based sports, ``team chemistry" is also accounted for in  \cite{DBLP:conf/asunam/DeLongTS13,DBLP:conf/pakdd/DeLongS12}.

\textbf{Recommendation and Expert Finding.} Recommendation and expert finding is a very active research topic in data mining and information retrieval, either to recommend products a user is mostly interested in or to identify the most knowledgeable people in a field. Our work is related to this in the sense that we aim to recommend top candidates who are most suitable for the vacancy. A popular method in recommendation (collaborative filtering) is latent factor model ~\cite{DBLP:journals/computer/KorenBV09, DBLP:journals/pieee/DrorKK12,DBLP:conf/kdd/WangB11}. The basic idea is to apply matrix factorization to user-item rating data to identify the latent factors. The factorization technique can be naturally extended by adding biases, temporal dynamics and varying confidence levels. In question-answering sites, \emph{e.g.}, Quora and Stack Overflow, an important task is to route a newly posted question to the `right' user with appropriate expertise and several methods based on link analysis have been proposed~\cite{DBLP:conf/www/ZhangAA07,DBLP:conf/kdd/BouguessaDW08,DBLP:conf/cikm/ZhouLLZ12}. In academia, identifying experts in a research field is of great value, \emph{e.g.}, assigning papers to the right reviewers in a peer-review process~\cite{DBLP:conf/kdd/MimnoM07,DBLP:conf/cikm/KarimzadehganZ09}, which can be done by either building the co-author network~\cite{DBLP:conf/www/LiTZLLH07} or using language model and topic-based model~\cite{DBLP:conf/icdm/DengKL08,DBLP:conf/cikm/HashemiNB13}. For enterprises, finding the desired specialist can greatly reduce costs and facilitate the ongoing projects. Many methods have been proposed to expert search through an organization's document repository~\cite{DBLP:conf/sigir/BalogAR06,DBLP:conf/wsdm/WuST13}.

\textbf{Graph Kernel.} Graph kernel measures the similarity between two graphs. Typical applications include automated reasoning~\cite{DBLP:conf/sdm/TsivtsivadzeUGH11}, bioinformatics/chemoinformatics~\cite{DBLP:conf/nips/FeragenKPBB13,DBLP:journals/jmlr/ShervashidzeSLMB11}. Generally speaking, graph kernels can be categorized into three classes: kernels based on walks~\cite{DBLP:conf/colt/GartnerFW03,DBLP:conf/nips/VishwanathanBS06, vishwanathan:graphkernels,GartnerLF04,BorgwardtK05}, kernels based on limited-sized subgraphs~\cite{HorvathGW04,ShervashidzeVPMB09,KondorSB09} and kernels based on subtree patterns~\cite{MaheUAPV05,ShervashidzeB09,HidoK09}. Graph kernels based on random walk is one of the most successful choices~\cite{BorgwardtKV07}. The idea is to perform simultaneous walks on the two graphs and count the number of matching walks. One challenge of random walk based graph kernel lies in computation. The straight-forward method needs $O(n^6)$ in time. For unlabelled graphs\hh{liangyue: double check if this is true}, the time complexity can be reduced to $O(n^3)$ by reducing to the problem of solving a linear system~\cite{DBLP:conf/nips/VishwanathanBS06, vishwanathan:graphkernels}. With low rank approximation, the computation can be further accelerated with high approximation accuracy~\cite{kang:fastwalk}.


%% file: 007con.tex
In this paper, we study the problem of \teamrep\ to recommend replacement when a critical team member becomes unavailable. To our best knowledge, we are the first to study this problem. The basic idea of our method is to adopt graph kernel to encode both {\em skill matching} and {\em structural matching}. To address the computational challenges, we propose a suite of fast  and scalable algorithms. Extensive experiments on real world datasets validate the effectiveness and efficiency of our algorithms. To be specific, (a) by bringing skill matching and structural matching together, our method is significantly better than the alternative choices in terms of both average precision (24\% better) and recall (27\% better);\hh{liangyue: fill in} and (b) our fast algorithms are orders of magnitude faster while enjoying a {\em sub-linear} scalability.

In the future, we would like to expand team replacement to team enhancement and team composition. For instance, given the structure of a high-grossing movie (e.g., {\em Saving Private Ryan}) of a particular genre, we want to develop effective algorithm to suggest a team of actors.

\hide{The main contributions of this paper are as follows:
\begin{itemize}
\item[1.] {\bf Problem Definition:} We introduce and formulate \teamrep.
\item[2.] {\bf Algorithm and Analysis:} We design exact and approximate algorithms and analyze their correctness and complexity.
\item[3.] {\bf Experimental Evaluations:} We perform extensive experiments on real world datasets to validate our algorithms' effectiveness and efficiency.
\end{itemize}

We only consider {\em single person replacement} in this paper, as a future direction, we will address the more challenging {\em multiple people replacement} problem. Note that the team size remains the same in the above problems. As another direction, we can study {\em team shrinkage} and {\em team expansion}.
}

%% file: paper.bbl
\begin{thebibliography}{45}

\bibitem[Anagnostopoulos et~al.(2012)Anagnostopoulos, Becchetti, Castillo,
  Gionis, and Leonardi]{DBLP:conf/www/AnagnostopoulosBCGL12}
A.~Anagnostopoulos, L.~Becchetti, C.~Castillo, A.~Gionis, and S.~Leonardi.
\newblock Online team formation in social networks.
\newblock In \emph{WWW}, pages 839--848, 2012.

\bibitem[Balog et~al.(2006)Balog, Azzopardi, and
  de~Rijke]{DBLP:conf/sigir/BalogAR06}
K.~Balog, L.~Azzopardi, and M.~de~Rijke.
\newblock Formal models for expert finding in enterprise corpora.
\newblock In \emph{SIGIR}, pages 43--50, 2006.

\bibitem[Bogdanov et~al.(2013)Bogdanov, Baumer, Basu, Bar-Noy, and
  Singh]{DBLP:conf/pkdd/BogdanovBBBS13}
P.~Bogdanov, B.~Baumer, P.~Basu, A.~Bar-Noy, and A.~K. Singh.
\newblock As strong as the weakest link: Mining diverse cliques in weighted
  graphs.
\newblock In \emph{ECML/PKDD (1)}, pages 525--540, 2013.

\bibitem[Borgwardt and Kriegel(2005)]{BorgwardtK05}
K.~M. Borgwardt and H.-P. Kriegel.
\newblock Shortest-path kernels on graphs.
\newblock In \emph{ICDM}, pages 74--81, 2005.

\bibitem[Borgwardt et~al.(2007)Borgwardt, Kriegel, Vishwanathan, and
  Schraudolph]{BorgwardtKV07}
K.~M. Borgwardt, H.-P. Kriegel, S.~V.~N. Vishwanathan, and N.~Schraudolph.
\newblock Graph kernels for disease outcome prediction from protein-protein
  interaction networks.
\newblock In \emph{Pacific Symposium on Biocomputing}, 2007.

\bibitem[Bouguessa et~al.(2008)Bouguessa, Dumoulin, and
  Wang]{DBLP:conf/kdd/BouguessaDW08}
M.~Bouguessa, B.~Dumoulin, and S.~Wang.
\newblock Identifying authoritative actors in question-answering forums: the
  case of yahoo! answers.
\newblock In \emph{KDD}, pages 866--874, 2008.

\bibitem[Cataldo and Ehrlich(2012)]{DBLP:conf/chi/CataldoE12}
M.~Cataldo and K.~Ehrlich.
\newblock The impact of communication structure on new product development
  outcomes.
\newblock In \emph{CHI}, pages 3081--3090, 2012.

\bibitem[Chang and Ehrlich(2007)]{DBLP:conf/cascon/ChangE07}
K.~Chang and K.~Ehrlich.
\newblock Out of sight but not out of mind?: Informal networks, communication
  and media use in global software teams.
\newblock In \emph{CASCON}, pages 86--97, 2007.

\bibitem[Cummings and Kiesler(2008)]{DBLP:conf/cscw/CummingsK08}
J.~N. Cummings and S.~B. Kiesler.
\newblock Who collaborates successfully?: prior experience reduces
  collaboration barriers in distributed interdisciplinary research.
\newblock In \emph{CSCW}, pages 437--446, 2008.

\bibitem[DeLong and Srivastava(2012)]{DBLP:conf/pakdd/DeLongS12}
C.~DeLong and J.~Srivastava.
\newblock Teamskill evolved: Mixed classification schemes for team-based
  multi-player games.
\newblock In \emph{PAKDD (1)}, pages 26--37, 2012.

\bibitem[DeLong et~al.(2013)DeLong, Terveen, and
  Srivastava]{DBLP:conf/asunam/DeLongTS13}
C.~DeLong, L.~G. Terveen, and J.~Srivastava.
\newblock Teamskill and the nba: applying lessons from virtual worlds to the
  real-world.
\newblock In \emph{ASONAM}, pages 156--161, 2013.

\bibitem[Deng et~al.(2008)Deng, King, and Lyu]{DBLP:conf/icdm/DengKL08}
H.~Deng, I.~King, and M.~R. Lyu.
\newblock Formal models for expert finding on dblp bibliography data.
\newblock In \emph{ICDM}, pages 163--172, 2008.

\bibitem[Dror et~al.(2012)Dror, Koenigstein, and
  Koren]{DBLP:journals/pieee/DrorKK12}
G.~Dror, N.~Koenigstein, and Y.~Koren.
\newblock Web-scale media recommendation systems.
\newblock \emph{Proceedings of the IEEE}, 100\penalty0 (9):\penalty0
  2722--2736, 2012.

\bibitem[Du et~al.(1999)Du, Faber, and Gunzburger]{du1999centroidal}
Q.~Du, V.~Faber, and M.~Gunzburger.
\newblock Centroidal voronoi tessellations: applications and algorithms.
\newblock \emph{SIAM review}, 41\penalty0 (4):\penalty0 637--676, 1999.

\bibitem[Feragen et~al.(2013)Feragen, Kasenburg, Petersen, de~Bruijne, and
  Borgwardt]{DBLP:conf/nips/FeragenKPBB13}
A.~Feragen, N.~Kasenburg, J.~Petersen, M.~de~Bruijne, and K.~M. Borgwardt.
\newblock Scalable kernels for graphs with continuous attributes.
\newblock In \emph{NIPS}, pages 216--224, 2013.

\bibitem[G{\"a}rtner et~al.(2003)G{\"a}rtner, Flach, and
  Wrobel]{DBLP:conf/colt/GartnerFW03}
T.~G{\"a}rtner, P.~A. Flach, and S.~Wrobel.
\newblock On graph kernels: Hardness results and efficient alternatives.
\newblock In \emph{COLT}, pages 129--143, 2003.

\bibitem[G{\"a}rtner et~al.(2004)G{\"a}rtner, Lloyd, and Flach]{GartnerLF04}
T.~G{\"a}rtner, J.~W. Lloyd, and P.~A. Flach.
\newblock Kernels and distances for structured data.
\newblock \emph{Machine Learning}, 57\penalty0 (3):\penalty0 205--232, 2004.

\bibitem[Golub and Van~Loan(1996)]{golub1996matrix}
G.~Golub and C.~Van~Loan.
\newblock Matrix computations.
\newblock 1996.

\bibitem[Hashemi et~al.(2013)Hashemi, Neshati, and
  Beigy]{DBLP:conf/cikm/HashemiNB13}
S.~H. Hashemi, M.~Neshati, and H.~Beigy.
\newblock Expertise retrieval in bibliographic network: a topic dominance
  learning approach.
\newblock In \emph{CIKM}, pages 1117--1126, 2013.

\bibitem[Hido and Kashima(2009)]{HidoK09}
S.~Hido and H.~Kashima.
\newblock A linear-time graph kernel.
\newblock In \emph{ICDM}, pages 179--188, 2009.

\bibitem[Hinds et~al.(2000)Hinds, Carley, Krackhardt, and
  Wholey]{Hinds00choosingwork}
P.~J. Hinds, K.~M. Carley, D.~Krackhardt, and D.~Wholey.
\newblock Choosing work group members: Balancing similarity, competence, and
  familiarity.
\newblock In \emph{Organizational Behavior and Human Decision Processes}, pages
  226--251, 2000.

\bibitem[Horv{\'a}th et~al.(2004)Horv{\'a}th, G{\"a}rtner, and
  Wrobel]{HorvathGW04}
T.~Horv{\'a}th, T.~G{\"a}rtner, and S.~Wrobel.
\newblock Cyclic pattern kernels for predictive graph mining.
\newblock In \emph{KDD}, pages 158--167, 2004.

\bibitem[Kang et~al.(2012)Kang, Tong, and Sun]{kang:fastwalk}
U.~Kang, H.~Tong, and J.~Sun.
\newblock Fast random walk graph kernel.
\newblock In \emph{SDM}, pages 828--838, 2012.

\bibitem[Karimzadehgan and Zhai(2009)]{DBLP:conf/cikm/KarimzadehganZ09}
M.~Karimzadehgan and C.~Zhai.
\newblock Constrained multi-aspect expertise matching for committee review
  assignment.
\newblock In \emph{CIKM}, pages 1697--1700, 2009.

\bibitem[Kondor et~al.(2009)Kondor, Shervashidze, and Borgwardt]{KondorSB09}
R.~I. Kondor, N.~Shervashidze, and K.~M. Borgwardt.
\newblock The graphlet spectrum.
\newblock In \emph{ICML}, page~67, 2009.

\bibitem[Koren et~al.(2009)Koren, Bell, and
  Volinsky]{DBLP:journals/computer/KorenBV09}
Y.~Koren, R.~M. Bell, and C.~Volinsky.
\newblock Matrix factorization techniques for recommender systems.
\newblock \emph{IEEE Computer}, 42\penalty0 (8):\penalty0 30--37, 2009.

\bibitem[Lappas et~al.(2009)Lappas, Liu, and Terzi]{DBLP:conf/kdd/LappasLT09}
T.~Lappas, K.~Liu, and E.~Terzi.
\newblock Finding a team of experts in social networks.
\newblock In \emph{KDD}, pages 467--476, 2009.

\bibitem[Li et~al.(2007)Li, Tang, Zhang, Luo, Liu, and
  Hong]{DBLP:conf/www/LiTZLLH07}
J.-Z. Li, J.~Tang, J.~Zhang, Q.~Luo, Y.~Liu, and M.~Hong.
\newblock Eos: expertise oriented search using social networks.
\newblock In \emph{WWW}, pages 1271--1272, 2007.

\bibitem[Mah{\'e} et~al.(2005)Mah{\'e}, Ueda, Akutsu, Perret, and
  Vert]{MaheUAPV05}
P.~Mah{\'e}, N.~Ueda, T.~Akutsu, J.-L. Perret, and J.-P. Vert.
\newblock Graph kernels for molecular structure-activity relationship analysis
  with support vector machines.
\newblock \emph{Journal of Chemical Information and Modeling}, 45\penalty0
  (4):\penalty0 939--951, 2005.

\bibitem[Mimno and McCallum(2007)]{DBLP:conf/kdd/MimnoM07}
D.~M. Mimno and A.~McCallum.
\newblock Expertise modeling for matching papers with reviewers.
\newblock In \emph{KDD}, pages 500--509, 2007.

\bibitem[Muller et~al.(2012)Muller, Ehrlich, Matthews, Perer, Ronen, and
  Guy]{DBLP:conf/chi/MullerEMPRG12}
M.~Muller, K.~Ehrlich, T.~Matthews, A.~Perer, I.~Ronen, and I.~Guy.
\newblock Diversity among enterprise online communities: collaborating,
  teaming, and innovating through social media.
\newblock In \emph{CHI}, pages 2815--2824, 2012.

\bibitem[Pei et~al.(2006)Pei, Wang, Liu, Wang, Wang, and
  Yu]{DBLP:journals/tkde/PeiWLWWY06}
J.~Pei, H.~Wang, J.~Liu, K.~Wang, J.~Wang, and P.~S. Yu.
\newblock Discovering frequent closed partial orders from strings.
\newblock \emph{IEEE Trans. Knowl. Data Eng.}, 18\penalty0 (11):\penalty0
  1467--1481, 2006.

\bibitem[Rangapuram et~al.(2013)Rangapuram, B{\"u}hler, and
  Hein]{DBLP:conf/www/RangapuramBH13}
S.~S. Rangapuram, T.~B{\"u}hler, and M.~Hein.
\newblock Towards realistic team formation in social networks based on densest
  subgraphs.
\newblock In \emph{WWW}, pages 1077--1088, 2013.

\bibitem[Shervashidze and Borgwardt(2009)]{ShervashidzeB09}
N.~Shervashidze and K.~M. Borgwardt.
\newblock Fast subtree kernels on graphs.
\newblock \emph{NIPS}, 2009.

\bibitem[Shervashidze et~al.(2009)Shervashidze, Vishwanathan, Petri, Mehlhorn,
  and Borgwardt]{ShervashidzeVPMB09}
N.~Shervashidze, S.~V.~N. Vishwanathan, T.~Petri, K.~Mehlhorn, and K.~M.
  Borgwardt.
\newblock Efficient graphlet kernels for large graph comparison.
\newblock \emph{Journal of Machine Learning Research - Proceedings Track},
  5:\penalty0 488--495, 2009.

\bibitem[Shervashidze et~al.(2011)Shervashidze, Schweitzer, van Leeuwen,
  Mehlhorn, and Borgwardt]{DBLP:journals/jmlr/ShervashidzeSLMB11}
N.~Shervashidze, P.~Schweitzer, E.~J. van Leeuwen, K.~Mehlhorn, and K.~M.
  Borgwardt.
\newblock Weisfeiler-lehman graph kernels.
\newblock \emph{Journal of Machine Learning Research}, 12:\penalty0 2539--2561,
  2011.

\bibitem[Simon(1996)]{Simon:1996:SA:237774}
H.~A. Simon.
\newblock \emph{The Sciences of the Artificial (3rd Ed.)}.
\newblock MIT Press, Cambridge, MA, USA, 1996.
\newblock ISBN 0-262-69191-4.

\bibitem[Tsivtsivadze et~al.(2011)Tsivtsivadze, Urban, Geuvers, and
  Heskes]{DBLP:conf/sdm/TsivtsivadzeUGH11}
E.~Tsivtsivadze, J.~Urban, H.~Geuvers, and T.~Heskes.
\newblock Semantic graph kernels for automated reasoning.
\newblock In \emph{SDM}, pages 795--803, 2011.

\bibitem[Vishwanathan et~al.(2006)Vishwanathan, Borgwardt, and
  Schraudolph]{DBLP:conf/nips/VishwanathanBS06}
S.~V.~N. Vishwanathan, K.~M. Borgwardt, and N.~N. Schraudolph.
\newblock Fast computation of graph kernels.
\newblock In \emph{NIPS}, pages 1449--1456, 2006.

\bibitem[Vishwanathan et~al.(2010)Vishwanathan, Schraudolph, Kondor, and
  Borgwardt]{vishwanathan:graphkernels}
S.~V.~N. Vishwanathan, N.~N. Schraudolph, R.~Kondor, and K.~M. Borgwardt.
\newblock Graph kernels.
\newblock \emph{The Journal of Machine Learning Research}, 99:\penalty0
  1201--1242, 2010.

\bibitem[Wang and Blei(2011)]{DBLP:conf/kdd/WangB11}
C.~Wang and D.~M. Blei.
\newblock Collaborative topic modeling for recommending scientific articles.
\newblock In \emph{KDD}, pages 448--456, 2011.

\bibitem[Wu et~al.(2013)Wu, Sun, and Tang]{DBLP:conf/wsdm/WuST13}
S.~Wu, J.~Sun, and J.~Tang.
\newblock Patent partner recommendation in enterprise social networks.
\newblock In \emph{WSDM}, pages 43--52, 2013.

\bibitem[Zadeh et~al.(2011)Zadeh, Balakrishnan, Kiesler, and
  Cummings]{DBLP:conf/chi/ZadehBKC11}
R.~Zadeh, A.~D. Balakrishnan, S.~B. Kiesler, and J.~N. Cummings.
\newblock What's in a move?: normal disruption and a design challenge.
\newblock In \emph{CHI}, pages 2897--2906, 2011.

\bibitem[Zhang et~al.(2007)Zhang, Ackerman, and
  Adamic]{DBLP:conf/www/ZhangAA07}
J.~Zhang, M.~S. Ackerman, and L.~A. Adamic.
\newblock Expertise networks in online communities: structure and algorithms.
\newblock In \emph{WWW}, pages 221--230, 2007.

\bibitem[Zhou et~al.(2012)Zhou, Lai, Liu, and Zhao]{DBLP:conf/cikm/ZhouLLZ12}
G.~Zhou, S.~Lai, K.~Liu, and J.~Zhao.
\newblock Topic-sensitive probabilistic model for expert finding in question
  answer communities.
\newblock In \emph{CIKM}, pages 1662--1666, 2012.

\end{thebibliography}
